%% file: paper.tex
\documentclass[10pt,journal,compsoc]{IEEEtran}
\ifCLASSOPTIONcompsoc
  \usepackage[nocompress]{cite}
\else
  \usepackage{cite}
\fi
\ifCLASSINFOpdf
\else
\fi
\usepackage{amsmath}
\usepackage{algorithmic}
\usepackage[tight,footnotesize]{subfigure}
\usepackage{url}
\usepackage{amsthm}
\newtheorem{theorem}{Theorem}
\newtheorem{lemma}[theorem]{Lemma}
\usepackage{adjustbox}
\usepackage{algorithm}
\usepackage{color}
\usepackage{multirow}
\usepackage{setspace}
\usepackage{enumitem, mathtools}

\begin{document}
\title{Optimizing Lossy Compression Rate-Distortion from Automatic Online Selection between SZ and ZFP}

\author{Dingwen Tao,~\IEEEmembership{Member,~IEEE,}
        Sheng Di,~\IEEEmembership{Senior,~IEEE,}
        Xin Liang,
        Zizhong Chen,~\IEEEmembership{Senior,~IEEE}
        and Franck Cappello,~\IEEEmembership{Fellow,~IEEE}
\IEEEcompsocitemizethanks{\IEEEcompsocthanksitem Dingwen Tao is with the Department
of Computer Science, The University of Alabama, Tuscaloosa, AL 35487.
\IEEEcompsocthanksitem Sheng Di and Franck Cappello are with the Mathematics and Computer Science division at Argonne National Laboratory, Lemont, IL 60439.
\IEEEcompsocthanksitem Xin Liang and Zizhong Chen are with the Department of Computer Science and Engineering at University of California, Riverside, CA 92521.
}}


\markboth{TRANSACTIONS ON PARALLEL AND DISTRIBUTED SYSTEMS, VOL. XX, NO. X, MONTH 2018}%
{Shell \MakeLowercase{\textit{Tao et al.}}: Optimizing Lossy Compression Rate-Distortion form Automatic Online Selection between SZ and ZFP}



\IEEEtitleabstractindextext{%
\input{tex/abstract}

}

\maketitle

% is not selected <OR> if conference mode is selected - because compsoc
% conference papers position the abstract like regular (non-compsoc)
\IEEEdisplaynontitleabstractindextext

% \fi
\IEEEpeerreviewmaketitle



\input{tex/intro}
\input{tex/related}
\input{tex/arch}
\input{tex/model}
\input{tex/evaluation}
\input{tex/conclusion}
\input{tex/acknowledge}

\ifCLASSOPTIONcaptionsoff
  \newpage
\fi



%\bibliographystyle{IEEEtran}
\bibliographystyle{IEEEtran}
\bibliography{refs}

\vspace{-5mm}


\begin{IEEEbiography}[{\includegraphics[width=1in,height=1.25in,clip,keepaspectratio]{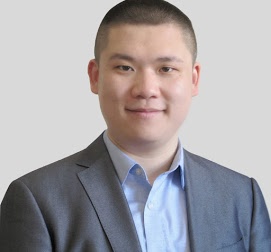}}]{Dingwen Tao}
Dingwen wen Tao received his bachelor's degree in mathematics from University of Science and Technology of China in 2013 and his Ph.D. degree in computer science from University of California, Riverside in 2018. He is currently an assistant professor of computer science in The University of Alabama. Prior to this, he worked at Brookhaven National Laboratory, Argonne National Laboratory, and Pacific Northwest National Laboratory. His research interests include high-performance computing, parallel and distributed systems, big data analytics, resilience and fault tolerance, scientific data compression, large-scale machine learning, numerical algorithms and software. He has published 15+ peer-reviewed papers in top HPC and big data conferences and journals, such as IEEE BigData, IEEE IPDPS, IEEE TPDS, ACM HPDC, ACM PPoPP, ACM/IEEE SC. Email: tao@cs.ua.edu.
\end{IEEEbiography}

\vspace{-5mm}

\begin{IEEEbiography}[{\includegraphics[width=1in,height=1.25in,clip,keepaspectratio]{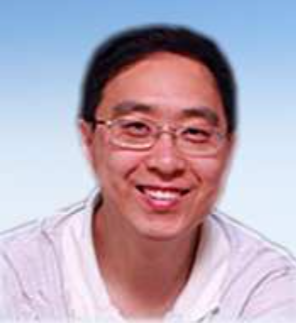}}]{Sheng Di}
Sheng Di received his master's degree from Huazhong University of Science and Technology in 2007 and Ph.D. degree from the University of Hong Kong in 2011. He is currently an assistant computer scientist at Argonne National Laboratory. Dr. Di's research interest involves resilience on high-performance computing (such as silent data corruption, optimization checkpoint model, and in-situ data compression) and broad research topics on cloud computing (including optimization of resource allocation, cloud network topology, and prediction of cloud workload/hostload). He is working on multiple HPC projects, such as detection of silent data corruption, characterization of failures and faults for HPC systems, and optimization of multilevel checkpoint models. Email: sdi1@anl.gov.
\end{IEEEbiography}

\vspace{-5mm}

\begin{IEEEbiography}[{\includegraphics[width=1in,height=1.25in,clip,keepaspectratio]{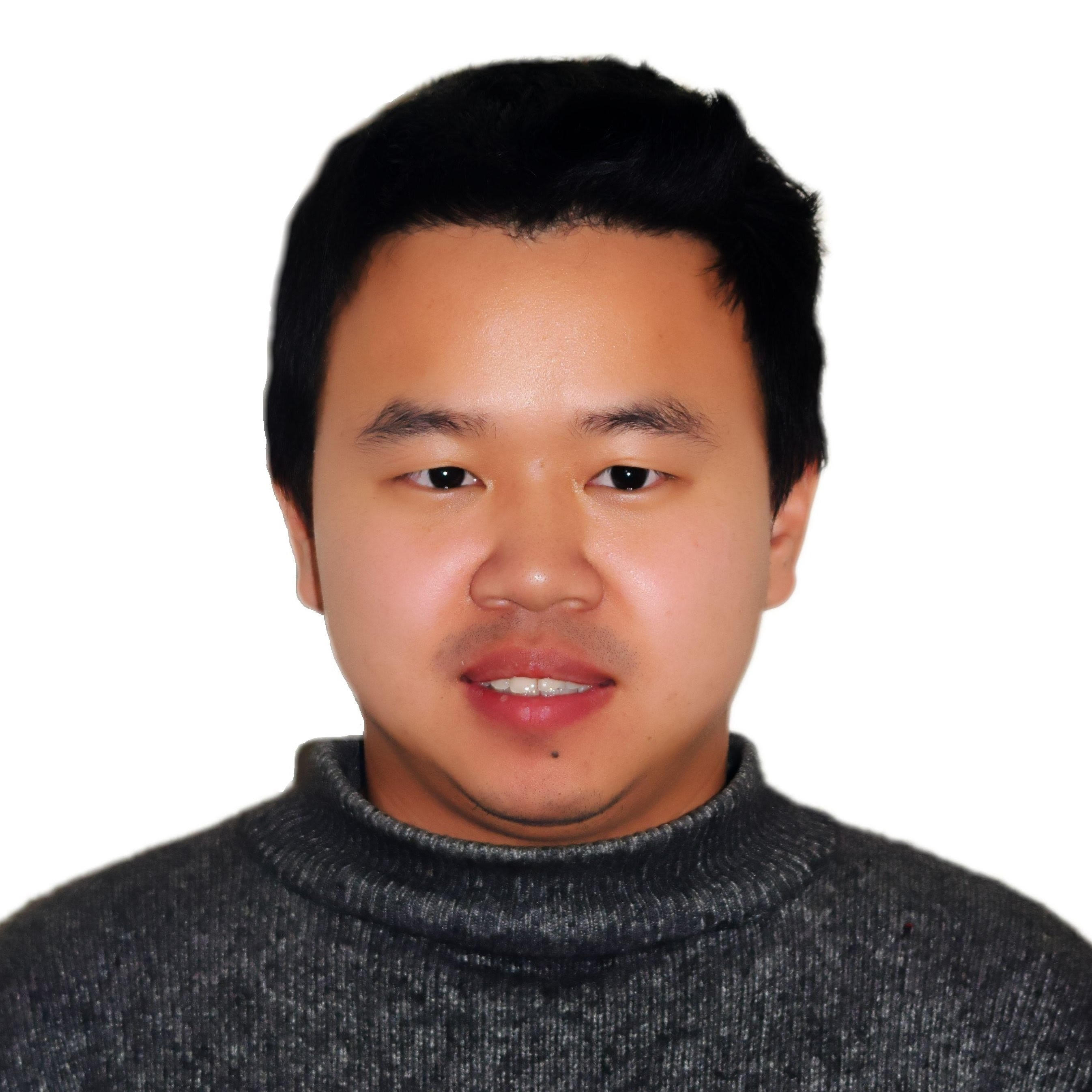}}]{Xin Liang}
Xin Liang received his bachelor's degree from Peking University in 2014 and will receive his Ph.D. degree from University of California, Riverside in 2019. His research interest includes parallel and distributed systems, fault tolerance, high-performance computing, large-scale machine learning, big data analysis and quantum computing. Email: xlian007@ucr.edu.
\end{IEEEbiography}

\vspace{-5mm}

\begin{IEEEbiography}[{\includegraphics[width=1in,height=1.25in,clip,keepaspectratio]{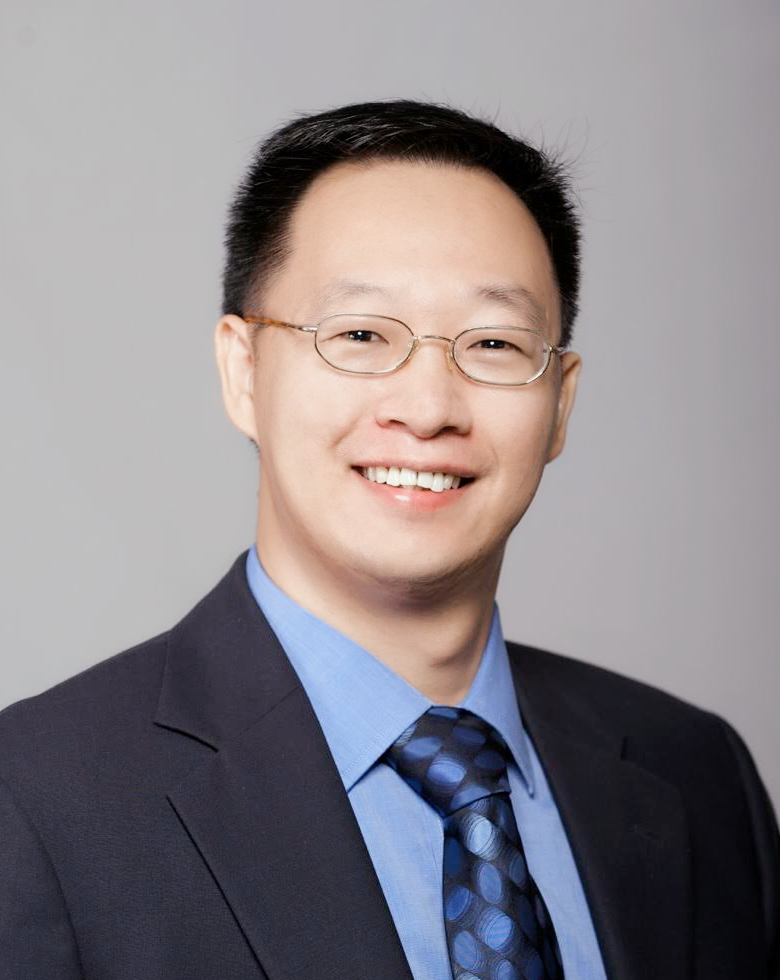}}]{Zizhong Chen}
Zizhong Chen received a bachelor's degree in mathematics from Beijing Normal University, a master's degree degree in economics from the Renmin University of China, and a Ph.D. degree in computer science from the University of Tennessee, Knoxville. He is an associate professor of computer science at the University of California, Riverside. His research interests include high-performance computing, parallel and distributed systems, big data analytics, cluster and cloud computing, algorithm-based fault tolerance, power and energy efficient computing, numerical algorithms and software, and large-scale computer simulations. His research has been supported by National Science Foundation, Department of Energy, CMG Reservoir Simulation Foundation, Abu Dhabi National Oil Company, Nvidia, and Microsoft Corporation. He received a CAREER Award from the US National Science Foundation and a Best Paper Award from the International Supercomputing Conference. He is a Senior Member of the IEEE and a Life Member of the ACM. He currently serves as a subject area editor for \textit{Elsevier Parallel Computing} journal and an associate editor for the \textit{IEEE Transactions on Parallel and Distributed Systems}. Email: chen@cs.ucr.edu.
\end{IEEEbiography}

\vspace{-5mm}

\begin{IEEEbiography}[{\includegraphics[width=1in,height=1.25in,clip,keepaspectratio]{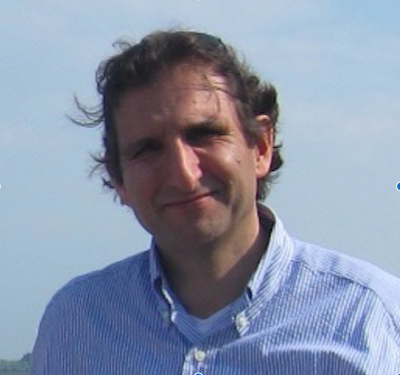}}]{Franck Cappello}
Franck Cappello is the director of the Joint-Laboratory on Extreme Scale Computing gathering six of the leading high-performance computing institutions in the world: Argonne National Laboratory, National Center for Scientific Applications, Inria, Barcelona Supercomputing Center, Julich Supercomputing Center, and Riken AICS. He is a senior computer scientist at Argonne National Laboratory and an adjunct associate professor in the Department of Computer Science at the University of Illinois at Urbana-Champaign. He is an expert in resilience and fault tolerance for scientific computing and data analytics. Recently he started investigating lossy compression for scientific data sets to respond to the pressing needs of scientist performing large-scale simulations and experiments. His contribution to this domain is one of the best lossy compressors for scientific data set respecting user-set error bounds. He is a member of the editorial board of the \textit{IEEE Transactions on Parallel and Distributed Computing} and of the \textit{ACM HPDC} and \textit{IEEE CCGRID} steering committees. He is a fellow of the IEEE. Email: cappello@mcs.anl.gov.
\end{IEEEbiography}


%\vfill


\end{document}

%% file: tex/abstract.tex
\begin{abstract}
\linespread{1.0}\selectfont
With ever-increasing volumes of scientific data produced by high-performance computing applications, significantly reducing data size is critical because of limited capacity of storage space and potential bottlenecks on I/O or networks in writing/reading or transferring data. SZ and ZFP are two leading BSD licensed open source C/C++ libraries for compressed floating-point arrays that support high throughput read and write random access.
However, their performance is not consistent across different data sets and across different fields of some data sets, which raises the need for an automatic online (during compression) selection between SZ and ZFP, with minimal overhead. In this paper, the automatic selection optimizes the rate-distortion, an important statistical quality metric based on the signal-to-noise ratio. To optimize for rate-distortion, we investigate the principles of SZ and ZFP. We then propose an efficient online, low-overhead selection algorithm that predicts the compression quality accurately for two compressors in early processing stages and selects the best-fit compressor for each data field. We implement the selection algorithm into an open-source library, and we evaluate the effectiveness of our proposed solution against plain SZ and ZFP in a parallel environment with 1,024 cores. Evaluation results on three data sets representing about 100 fields show that our selection algorithm improves the compression ratio up to 70\% with the same level of data distortion because of very accurate selection (around 99\%) of the bestfit compressor, with little overhead (less than 7\% in the experiments).
\end{abstract}

%% file: tex/intro.tex
\section{Introduction}
\label{sec:introduction}
\IEEEPARstart{A}{n} efficient scientific data compressor is increasingly critical to the success of today's scientific research because of the extremely large volumes of data produced by today's high-performance computing (HPC) applications. The Community Earth Simulation Model (CESM) \cite{cesm,baker}, for instance, produces terabytes of data every day. In the Hardware/Hybrid Accelerated Cosmology Code (HACC) \cite{hacc} (a well-known cosmology simulation code), the number of particles to simulate could reach up to 3.5 trillion, which may produce 60 petabytes of data to store. On the one hand, such large volumes of data cannot be  stored even in a parallel file system (PFS) of a supercomputer, such as the Mira \cite{mira} supercomputer at Argonne because it has only 20 petabytes of storage space. On the other hand, the I/O bandwidth may also become a serious bottleneck. The memory of extreme-scale systems continues to grow, with a factor of 5 or more expected for the next generation of systems compared with the current one (e.g., the Aurora supercomputer \cite{aurora} has over 5 PB total memory); however, although the Burst Buffer technology \cite{burst} can relieve the I/O burden to some extent, the bandwidth of PFS  is still developing relatively slowly compared with the memory capacity and peak performance. Hence, storing application data to file systems for postanalysis will take much longer than in current systems.

Error-controlled lossy compressors for scientific data sets have been studied for years, because they not only significantly reduce data size but also keep decompressed data valid to users. The existing lossy compressors, however, exhibit largely different compression qualities depending on various data sets because of their different algorithms and the diverse features of scientific data. The atmosphere simulation (called ATM) in the CESM model, for example, has over 100 fields (i.e., variables), each of which may have largely different features.
We note that various variables or data sets work better with different compression techniques. For instance, SZ \cite{sz16,sz17,sz18} exhibits better compression quality than does ZFP \cite{zfp,ipdps18} on some data sets, whereas ZFP is better on others. With the same level of distortion of data\textemdash peak signal-to-noise ratio (PSNR), for instance\textemdash SZ exhibits a better compression ratio than does ZFP on 72.8\% of the fields in the ATM simulation data, while ZFP wins on the remaining 27.2\% fields. One key question is: Can we develop a lightweight online selection method that can estimate on the fly the best-fit compression technique for any data set, such that the overall compression quality can be improved significantly for that application?

In this work, we propose a novel online selection method for optimizing the error-controlled lossy compression quality of structured HPC scientific data in terms of rate-distortion, which is the first attempt to our knowledge. The lossy compression quality is assessed mainly by rate-distortion in the scientific data compression community \cite{sz17,zfp,zchecker}. However, designing an effective method that can select the best-fit compressor based on \textbf{rate-distortion} is challenging because rate-distortion is not a single metric for a given data set; rather it involves a series of compression cases with different data distortions and compression ratios. Hence, selecting the best-fit compressor based on rate-distortion by simply running SZ and ZFP once based on sampled data points is impossible. Unlike Lu et al.'s work \cite{ipdps18}, which selects the best compressor based on a specific error bound and sampled data points, we have to model accurately the principles of the two state-of-the-art lossy compressors both theoretically and empirically. This is nontrivial because of the diverse data features and multiple complicated compression stages involved in the two compressors. Moreover, we must assure that our estimation algorithm has little computation cost, in order to keep a high overall execution performance for the in situ compression.

The contributions of this paper are as follows.
\begin{itemize}
\item We conduct an in-depth analysis of the existing lossy compression techniques and divide the procedure of lossy compression into three stages: lossless transformation for energy compaction; lossy compression for data reduction; and lossless entropy encoding, which is a fundamental step for the compression-quality estimation.
\item We explore a series of efficient strategies to predict the compression quality (such as compression ratio and data distortion) for the two leading lossy compressors (i.e., SZ and ZFP) accurately. Specifically, we derive some formulas and approaches for accurate prediction of PSNR and the number of bits to represent a data value on average (i.e., bit-rate) based on in-depth analysis of their compression principles with the three compression stages.
\item Based on our compression-quality estimation, we develop a novel online method to select the best-fit compressor between SZ and ZFP for each data set, leading to the best lossy compression results. We adopt \emph{rate-distortion} as the selection criterion because it involves both compression ratio and data distortion and it has been broadly used to assess compression quality in many domains \cite{baker, zfp, sz17}.
%\item We implement our proposed method and release it as an open-source under a BSD license. Our developed toolkit is designed to be parallelizable especially for the in situ parallel compression during scientific simulation, in that the data compression does not need any communication across processes/ranks.
\item We evaluate the performance and compression quality of our proposed solution on a parallel system with 1,024 cores. Experiments on structured data sets from real-world HPC simulations show that our solution can significantly improve the compression ratio with the same level of data distortion and comparable performance. The compression ratio can be improved by up to 70\% because of a very high accuracy (around 99\%) in selecting the best compressor in our method. With our solution, the overall performance in loading and storing data can be improved by 79\% and 68\% on 1,024 cores, respectively, compared with the second-best evaluated approach.
\end{itemize}

The remainder of this paper is organized as follows. In Section \ref{sec:related-work}, we discuss the related work in scientific data compression. In Section \ref{sec:arch}, we introduce the overall architecture of our proposed automatic online selection method. In Section \ref{sec:analysis}, we analyze the three critical stages in detail based on existing lossy compressors. In Section \ref{sec:estimation}, we discuss how to predict the compression quality for SZ and ZFP accurately. In Section \ref{sec:exper-eval}, we present and analyze the experimental results. In Section \ref{sec:conclusion}, we provide concluding remarks and a brief discussion of future work.

\vspace{-1mm}

%% file: tex/related.tex
\section{Related Work}
\label{sec:related-work}

The issue of scientific data compression has been studied for years. The data compressors can be split into two categories: lossless compressor and lossy compressor. 

Lossless compressors make sure that reconstructed data set after decompression is exactly the same as the original data set. Such a constraint significantly limits their compression ratio on scientific data, for whatever generic byte-stream compressors (such as Gzip \cite{gzip} and bzip2 \cite{bzip2}) or whatever floating-point data compressors (such as FPC \cite{fpc} and FPZIP \cite{fpzip}). The reason is that scientific data is composed mainly of floating-point values and their tailing mantissa bits could be too random to compress effectively \cite{gomez2013improving}. 

Lossy compression techniques for scientific data sets generated by HPC applications also have been studied for years, and the existing state-of-the-art compressors include SZ \cite{sz16,sz17,sz18}, ZFP \cite{zfp}, ISABELA \cite{isabela}, FPZIP \cite{fpzip}, SSEM \cite{ssem}, VAPOR \cite{vapor}, and NUMARCK \cite{numarck}. Basically, their compression models can be summarized into two categories: prediction-based model \cite{sz17,isabela,fpzip,numarck,drbsd1} and transform-based model \cite{zfp,ssem}. A prediction-based compressor needs to predict data values for each data point and encodes the difference between every predicted value and its corresponding real value based on a quantization method. Typical examples are SZ \cite{sz16,sz17,sz18}, ISABELA \cite{isabela}, and FPZIP \cite{fpzip}. The block transform-based compressor transforms original data to another space where the majority of the generated data are very small (close to zero), such that they can be stored with a certain loss in terms of the user's required error bound. For instance, JPEG \cite{jpeg}, SSEM \cite{ssem} and VAPOR \cite{vapor}, and ZFP \cite{zfp} adopt discrete cosine transform, discrete wavelet transform and a customized orthogonal transform, respectively.

Recently, many research studies \cite{ipdps18,huebl2017scalability,sz16,sz17,sz18,zfp,zfp-online,peter-distribution,liang2018efficient,gok2018pastri,tao2017depth} have showed that SZ and ZFP are two leading lossy compressors for HPC scientific data. Specifically, SZ predicts each data point's value by its preceding neighbors in the multidimensional space and then performs an error-controlled quantization and customized Huffman coding to shrink the data size significantly. ZFP splits the whole data set into many small blocks with an edge size of 4 along each dimension and compresses the data in each block separately by a series of carefully designed steps, including alignment of exponent, orthogonal transform, fixed-point integer conversion, and bit-plane-based embedded coding. For more details, we refer readers to \cite{sz17} and \cite{zfp} for SZ and ZFP, respectively. Fu et al. \cite{otf} proposed an on-the-fly lossy compression method for a high-performance earthquake simulation. Their lossy compression can reduce the memory cost  by 50\% and improve the overall performance 24\% on Sunway TaihuLight supercomputer\mbox{\cite{sunway}}. This on-the-fly lossy compression scheme reduces 32-bit floating-point data to 16-bit by using an adaptive binary representation. We exclude this approach in our work because it is limited compression ratio of 2.

How to integrate different compression techniques into one framework and use their distinct advantages to optimize the compression quality is a challenging topic. Blosc \cite{blosc} is a successful lossless compressor based on multiple different lossless compression methods, including FastLZ \cite{fastlz}, LZ4/LZ4HC \cite{lz4}, Snappy \cite{snappy}, Zlib \cite{zlib}, and Zstandard \cite{zstd}. However, no such a compressor has been designed based on different lossy compression techniques for optimizing the rate-distortion, leading to a huge gap for the demand of efficient error-controlled compression on scientific data sets. Although Lu et al. \cite{ipdps18} model the compression performance of SZ and ZFP to select the best compressor in-between, they focus only on the compression ratio given a specific maximum error. Their method does not model and select the best compressor based on statistical metrics such as RMSE (root mean square error) or PSNR (a simple derivation from RMSE). Such average error metrics are more important for visualization of scientific data than is the maximum error \cite{baker,woodring2011revisiting}.
Although there are some more complex metrics (such as Structural Similarity Index) that can also evaluate compression schemes in terms of visual quality, for generality and simplicity \cite{hore2010image}, our research focuses on \textit{maximizing PSNR} for a given compression ratio, by modeling and selecting online the best compressor between SZ and ZFP with low performance overhead. 

%% file: tex/arch.tex
\section{Architecture of Proposed Online Automatic Selection Method for Lossy Compression}
\label{sec:arch}
Lossy compression can be divided into three stages as shown in Figure \ref{fig:f8}.
In Stage I, the original data is transformed to another data domain (e.g., frequency domain) by lossless transformations. 
Here lossless transformation means the reconstructed data will be lossless if one transforms the original data and does the corresponding inverse-transformation right away. 
The transformed data is easier to compress because of the efficient energy compaction \cite{malvar1989lot}. Energy compaction means that the energy is more concentrated in some elements of the transformed data compared to the distribution of energy in the original data.
Stage II reduces the data size but also introduces errors. The most commonly used techniques for Stage II are vector quantization \cite {vq} (static quantization) and embedded coding \cite{ec} (dynamic quantization).
Stage III performs entropy coding for further lossless data reduction, and it is sometimes optional. Figure \ref{fig:f8} shows four state-of-the-art lossy compressors for HPC scientific data and their corresponding techniques in each stage.

\begin{figure}[ht]
\centering
\includegraphics[scale=0.34]{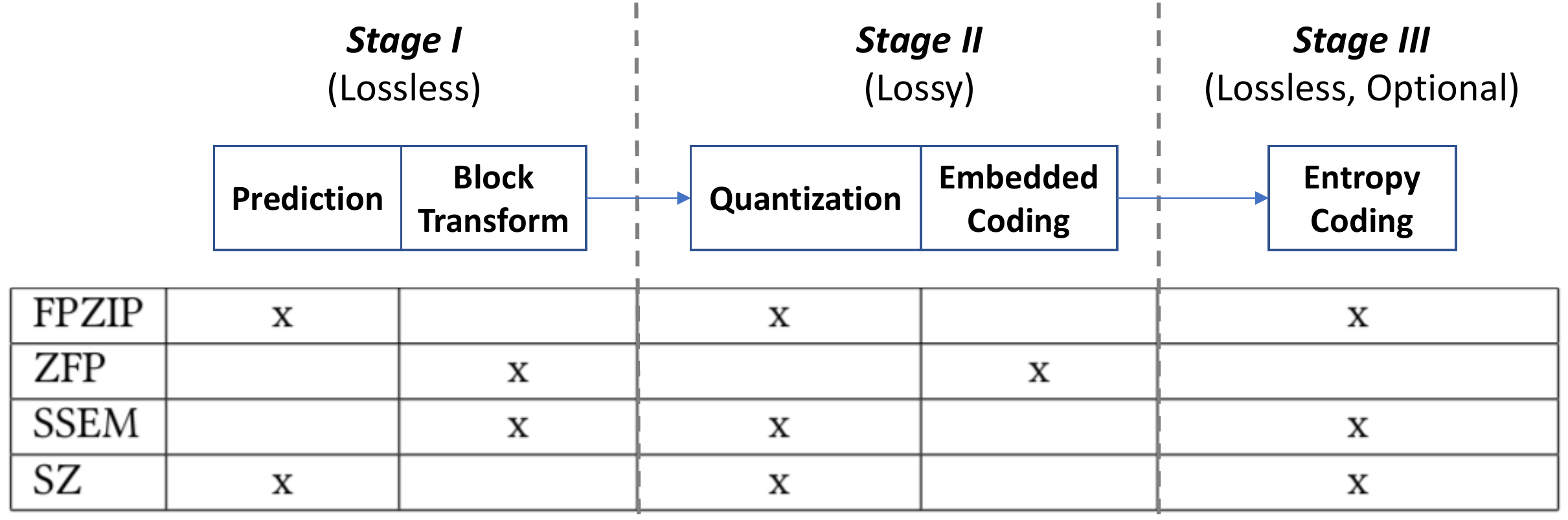}
\caption{Three stages in lossy compression for HPC scientific data.}
\label{fig:f8}
\end{figure}

We design our online selection method based on the analysis of the three critical compression stages.
Specifically, we propose a novel optimization strategy comprising four steps as shown in Figure \ref{fig:workflow}. The first step takes the input scientific data sets and performs Stage I's transformation on the sampled data points. The second step uses the transformed data points from Step 1 to estimate the compression quality (including compression ratio and distortion of data) based on our proposed estimation model. The third step selects the best-fit lossy compression strategy based on SZ and ZFP. The fourth step constructs a lossy compressor and uses it for compressing the data set.
As confirmed by recent research \cite{ipdps18,sz17,zfp,zfp-online,drbsd1, drbsd2,huebl2017scalability}, SZ and ZFP are two leading lossy compressors for HPC scientific data and can well represent prediction-based and transformation-based lossy compressors, respectively. Accordingly, our online selection method mainly is based on these two state-of-the-art lossy compressors without loss of generality.

In this paper, we adopt a practical \textit{in situ model} \cite{insitu} as many state-of-the-art lossy compressors (such as ISABELA, SZ, and ZFP) use as well. Here in situ means data analysis, visualization, and compression happen without first writing data to persistent storage. Thus, in the in situ compression model, the compression will be conducted after the entire computation in each simulation time step such that the entire analysis data is already kept in memory. Overall, our online method can select the best-fit compressor on-the-fly during the in situ compression, aiming at reducing data storage and I/O time overheads.

\begin{figure}
\centering
\includegraphics[scale=0.38]{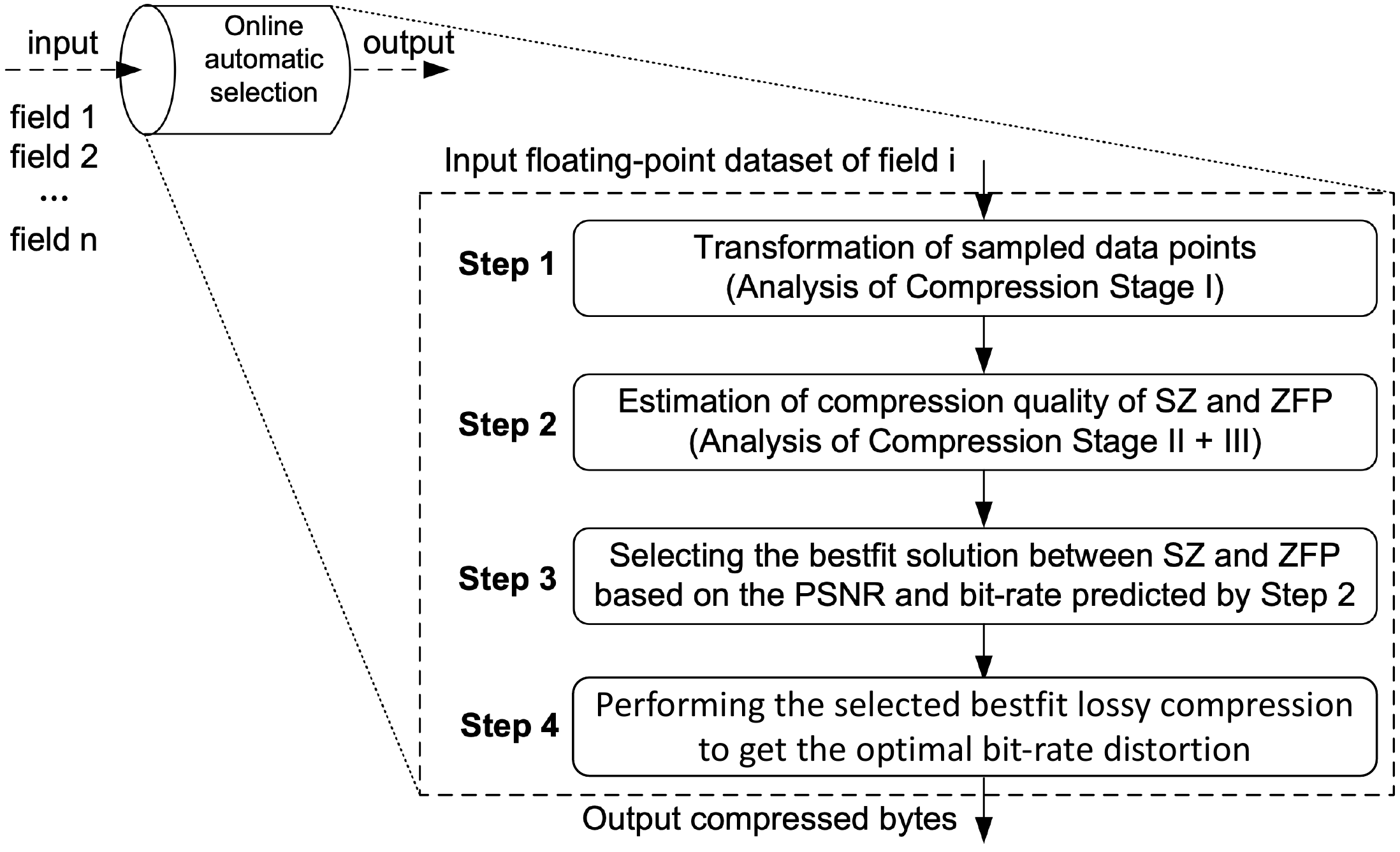}
\caption{Workflow of proposed online, low-overhead selection method for lossy compression of HPC scientific data.}
\label{fig:workflow}
\vspace{-4mm}
\end{figure}

%% file: tex/model.tex
\section{Analysis of Lossless Transformations for Energy Compaction in Stage I}
\label{sec:analysis}

In this section, we provide an in-depth analysis of the impact of Stage I (i.e., \emph{prediction-based transformation} (PBT) and \emph{block orthogonal transformation} (BOT)) on the overall distortion of data.
As presented in Figure \ref{fig:f8}, Stage I is lossless.
However, this does not mean that if the data in the transformed domain is changed, the overall distortion level (such as mean squared error (MSE)) of the finally reconstructed data can stay the same as that of the transformed data. The reason is that the data in the new transformed domain will be largely different from the original data. 
Based on an in-depth analysis of the two transformation methods in Stage I, we prove that the $L^2$-norm-based error value (e.g, MSE) keeps unchanged after the inverse transformation of PBT and BOT.
This fundamental analysis implies that we can predict the overall distortion of the finally decompressed data for SZ and ZFP by estimating the data distortion in Stage II.

\subsection{Prediction-Based Transformation (PBT)}
\label{sec:pbt}
In this subsection, we introduce prediction-based lossy compression, and then we infer that the pointwise compression error (i.e., the difference between any original data value and its decompressed value) is equal to the error introduced by vector quantization or embedded encoding in Stage II.

\begin{figure}
\centering
\includegraphics[scale=0.44]{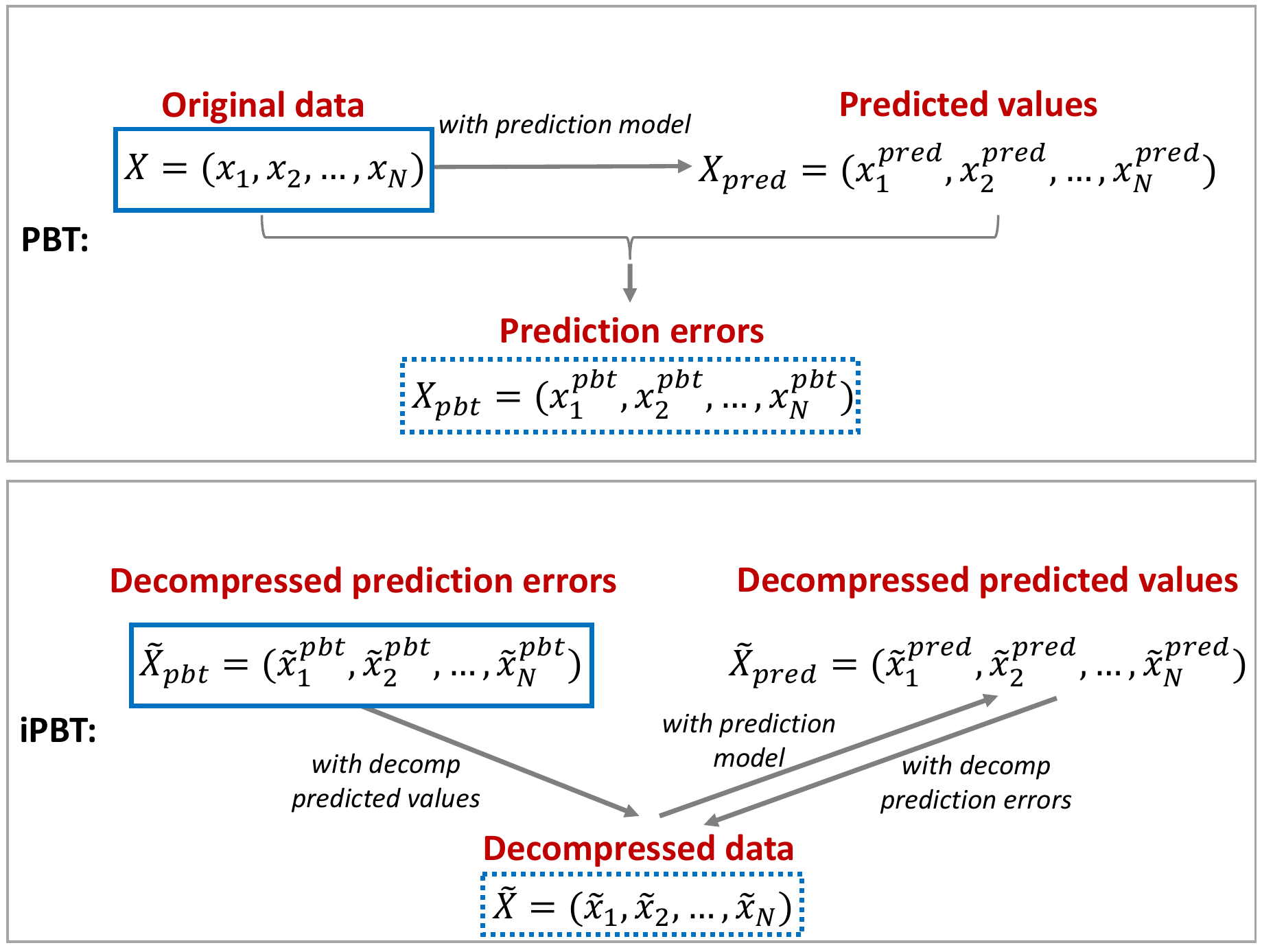}
\caption{Prediction-based transformation (PBT) and inverse prediction-based transformation (iPBT).}
\label{fig:f4}
\vspace{-6mm}
\end{figure}

In the compression phase of the prediction-based lossy compression (as shown in the top subfigure of Figure \ref{fig:f4}), the first step is to predict the value of each data point and calculate the prediction errors. We define \textit{PBT} as the process of generating a set of prediction errors (denoted by $X_{pbt}$) based on the original data set (denoted by $X$), data point by data point during the compression.
The prediction error will be quantized or encoded in Stage II.

During the decompression, one needs to reconstruct the prediction errors based on quantization method or embedded encoding and then reconstruct the overall data set by an inverse PBT (as presented in the bottom subfigure of Figure \ref{fig:f4}). We define the \textit{inverse PBT} (denoted \textit{iPBT}) as the procedure of constructing the decompressed data set (denoted $\tilde{X}$), data point by data point, based on the reconstructed prediction errors (denoted $\tilde{X}_{pbt}$) during the decompression.

In what follows, we infer that the following equation must hold for PBT.
\begin{equation}
\label{eq:x1}
X - \tilde{X} = X_{pbt} - \tilde{X}_{pbt}
\end{equation}

During the compression, the prediction method generally predicts the value of each data point based on the data points nearby in space because of the potential high consecutiveness of the data set. The Lorenzo predictor \cite{lorenzo}, for example, approximates each data point by the values of its preceding adjacent data points. \footnote{Lorenzo predictor uses 1 neighbor per data point for 1D data, 3 neighbors per data point for 2D data, and 7 neighbors per data point for 3D data.} Since the neighboring data values to be used to reconstruct each data point during the decompression are actually the decompressed values instead of the original values, in practice, one has to assure that the compression and decompression stage have exactly the same prediction procedure (including the data values used in the prediction method); 
otherwise, the data loss will be propagated during the decompression. Hence, the predicted values during the compression must be equal to the predicted values during the decompression. That is, we have $X_{pred} = \tilde{X}_{pred}$. Then, we can derive Equation (\ref{eq:x1}) based on the following two equations: $X_{pbt}=X-X_{pred}$ and $\tilde{X}=\tilde{X}_{pbt} + \tilde{X}_{pred}$.

Based on Equation (\ref{eq:x1}), we can easily derive the following theorem.
\begin{theorem}
\label{thm:2}
The pointwise compression error in the original data space is the same as the pointwise compression error in the PBT-transformed data space.
\end{theorem}

\subsection{Block Orthogonal Transformation (BOT)}
\label{sec:bot}
In the following discussion, we first introduce the principle of the block orthogonal transformation. We then prove a critical feature: \emph{the $L_2$-norm based compression error (such as MSE) in the original data space is the same as the compression error in the BOT transformed data space}.

Let us first describe the elementwise tensor (matrix) norms that we will use in the following discussion. One can treat a tensor as a vector and calculate its elementwise norm based on a specific vector norm. For example, by using vector $p$-norm, we can define the elementwise $L_{p}$ norm of a tensor $X$ as follow. 
\begin{align}
\resizebox{.58\hsize}{!}{$
\|X\|_p = \|vec(X)\|_p = (\sum\limits_{x \in X} x^p)^{1/p}$} \label{eq:normp}
\end{align}

Further, if $X$ is an $M \times M$ matrix and we choose $p = 2$, Equation (\ref{eq:normp}) can be simplified to
\begin{align}
\resizebox{.75\hsize}{!}{$\|X\|_2 = (\sum\limits_{i=1}^M \sum\limits_{j=1}^M x_{ij}^2)^{1/2} = \sqrt{trace(X^t \cdot X)}$}, \label{eq:norm2}
\end{align}
where $trace()$ returns the sum of diagonal entries of a square matrix. Equation (\ref{eq:norm2}) defines the elementwise $L_2$ norm (a.k.a. Frobenius norm) of a square matrix.

Block transformation-based lossy compressors divide the entire data set into multiple data blocks and perform blockwise transformation at Stage I. Unlike prediction-based transformation, each block transformation has no dependency and can be performed independently.
Each block transformation is composed of several 1D linear transformations that can be performed along each axis within the block.
For example, in a 2D data array, 1D linear transformation is applied to each row ($x$-axis) and each column ($y$-axis).
Each 1D linear transformation can be calculated as a multiplication of the transformation matrix and 1D vector.

Many lossy compressors adopt orthogonal matrices in their transformations.
For example, SSEM uses the Haar wavelet transform and ZFP uses a self-optimized orthogonal matrix.
Here an orthogonal matrix $T$ means its columns and rows are orthogonal unit vectors, i.e., $T \cdot T^t = I$, where $I$ is the identity matrix.
The most significant advantage of using orthogonal transformation 
is the property of \textit{$L_2$-norm invariance} after transformation, that is,
\begin{align}
\|T \cdot X\|_2 & = \sqrt{trace((T \cdot X)^t \cdot (T 
\cdot X))} \nonumber \\
& = \sqrt{trace(X^t \cdot T^t \cdot T \cdot X)} \nonumber \\
& = \sqrt{trace(X^t \cdot X)} = \|X\|_2. \label{eq:invariant}
\end{align}
Based on this property, we can prove that the $L_{2}$-norm-based compression error in the original data space is the same as the compression error in the BOT-transformed data space. We will prove it later in this subsection.

The block size in the BOT-based lossy compressor is usually set to the power of 2.
ZFP and SSEM, for example, set the block size to $4^n$, where $n$ is the dimension size ($n = 1, 2, 3$). JPEG uses $8\times8$ as the block size in 2D image data.
In our work, without loss of generality, we consider the block size in BOT to be $4^n$ and do not limit the dimension $n$ of the data set. Note that here the ``dimension'' represents the dimensionality of each data point rather than the number of fields in the data sets.

Based on prior research \cite{zfp}, the transformation matrix of most existing well-known BOTs can be expressed as a uniform parametric form as
\begin{equation*}
\small
T = \frac{1}{2}
\begin{pmatrix}
1 & 1 & 1 & 1 \\
c & s & -s & -c \\
1 & -1 & -1 & 1 \\
s & -c & c & -s
\end{pmatrix}
\end{equation*}
\begin{equation*}
s = \sqrt{2}\sin \frac{\pi}{2}t
\quad
c = \sqrt{2}\cos \frac{\pi}{2}t,
\end{equation*}
where $t \in [0,1]$ is a parameter.
Specifically, $t = 0$ and $t = \frac{1}{4}$ corresponds to discrete HWT and DCT II, respectively, which are two most common transforms.
Moreover, $t = \{ \frac{2}{\pi}\tan^{-1} \frac{1}{3}, \frac{2}{\pi}\tan^{-1}\frac{1}{2}, \frac{1}{2}\}$ represents slant transform, high-correlation transform, and Walsh-Hadamard transform. %respectively

In what follows, we discuss the unified formulas of BOT for any dimensional data.
We use $T_{bot}$ to denote the BOT and $X$ to denote the data block. 
$X$ can be represented by a $4^n$ tensor $(x_{i_1 \cdots i_n})_{4 \times \cdots \times 4}$ where $1 \leq  i_1, \cdots, i_n \leq 4$.
Since the orthogonal transformation is performed on the 1D $4 \times 1$ vector, we need to rearrange $X$'s elements to form an $4\times 4^{n-1}$ matrix and do a matrix-matrix multiplication.
The $n$ dimensional tensor $X$ can be unfolded along the $n$ directions by index mapping.
We use $D_1$-axis, $D_2$-axis, $\cdots$, $D_n$-axis to denote the $n$ directions.
Specifically, the unfolding along the $k$-th direction $D_k$-axis ($1 \leq k \leq n$) will map the tensor element $x_{i_1 \cdots i_n}$ to the matrix element $(i_k, j)$, where $j = \sum_{l=1}^{k-1} 4^{l-1} (i_l-1) + \sum_{l=k+1}^n 4^{l-2} (i_l-1)$. We use $unfold_{D_1}()$, $unfold_{D_2}()$, $\cdots$, $unfold_{D_n}()$ to denote the unfolding operations along the $D_1$-axis, $D_2$-axis, $\cdots$, $D_n$-axis, respectively.
Accordingly, we can fold the tensor from the unfolded matrix by the inverse index mapping, denoted by $fold_{D_1}()$, $fold_{D_2}()$, $\cdots$, $fold_{D_n}()$.
Thus, $T_ {bot}$ can be expressed as the following $n$ operations.

\begin{itemize}%[noitemsep,topsep=8pt]
\item[$1.$] $X = fold_{D_1} ( T \cdot unfold_{D_1}(X) )$
\item[$2.$] $X = fold_{D_2} ( T \cdot unfold_{D_1}(X) )$
\item[$\vdotswithin{n.}$] \leavevmode
\item[$n.$] $X = fold_{D_n} ( T \cdot unfold_{D_n}(X) )$ 
\end{itemize} 

Next, we propose Lemma \ref{lemma:1} and Theorem \ref{thm:1} and prove them.

\begin{lemma}
\label{lemma:1}
Block orthogonal transformation (BOT) preserves the $L_{2}$ norm on any dimenstional data sets.
\end{lemma}
\begin{proof}
We still denote the orthogonal transformation matrix by $T$.
Because the $unfold_{D_k}()$ and $fold_{D_k}()$ are both index mapping operations, the values and elementwise norm will remain unchanged. Thus, we can write 
\begin{align*}
\|fold_{D_k} (T \cdot unfold_{D_k}(X) )\|_{2} = \|T \cdot unfold_{D_k}(X)\|_{2}.
\end{align*}
Then, based on Equation (\ref{eq:invariant}), we can get
\begin{align*}
\|T \cdot unfold_{D_k}(X)\|_{2} = \|unfold_{D_k}(X)\|_{2} = \|X\|_{2}.
\end{align*}
So $\|fold_{D_k} (T \cdot unfold_{D_k}(X) )\|_{2} = \|X\|_{2}$ is held for the $k$-th operation ($1 \leq k \leq n$), which demonstrates that every operation in the BOT can keep $\|X\|_2$ unchanged. Therefore, we have proved this theorem.
\end{proof}

We still use $\tilde{X}$ to denote the decompressed block data, $X_{bot}$ to denote the transformed block data in the compression, and $\tilde{X}_{bot}$ to denote the decompressed transformed block data in the decompression.
We have 
\begin{equation*}
X_{bot} = T_{bot} (X)
\end{equation*}
and 
\begin{equation*}
\tilde{X}_{bot}= T_{bot} (\tilde{X}).
\end{equation*}
Thus, due to the linearity of $T_{bot}$, we have
\begin{equation*}
X_{bot}  - \tilde{X}_{bot} = T_{bot}(X) - T_{bot}(\tilde{X}) = T_{bot}(X - \tilde{X}).
\end{equation*}
Based on Lemma \ref{lemma:1}, we have 
\begin{equation*}
\|X_{bot}  - \tilde{X}_{bot}\|_{2} = \|T_{bot}(X - \tilde{X})\|_{2} = \|X - \tilde{X}\|_{2}.
\end{equation*}
This equation also holds when $X$ is composed of multiple data blocks.
That is, we already prove the following theorem.

\begin{theorem}
\label{thm:1}
The $L_{2}$-norm-based compression error in the original data space is the same as the compression error in the BOT-transformed data space on any dimenstional data sets.
\end{theorem}

Note that the reasons that Theorem \ref{thm:1} focuses on $L_2$ norm include two aspects: on one hand, the ``norm invariance'' property (as shown in Equation (\ref{eq:invariant})) of BOT only holds for $L_2$ norm in terms of the elementwise $L_p$ norms because of Equation \ref{eq:norm2}; on the other hand, $L_2$-norm-based error (such as MSE or PSNR) has been considered as one of the most critical indicators to assess the overall data distortion in literature, because it is closely related to the visual quality \cite{guthe2001real}, unlike maximum compression error (i.e., $L_{\inf}$-norm based error).

\subsection{Data Sampling for Compression-Quality Estimation}
In our proposed automatic online selection method, we first sample the data points (i.e., Step 1) and then perform the transformations on them (i.e., Step 2) in order to estimate the overall compression quality, as shown in Figure \ref{fig:workflow}. The distance between 
two data blocks sampled nearby will be fixed in the same dimension and different across dimensions, such that all sampled blocks can be distributed uniformly throughout the entire data set.
We use the term \textit{sampling rate} to represent the sampling frequency, which is denoted by $r_{sp}$.
In the evaluation section, we present the accuracy of our estimation model with respect to different sampling rates. Based on our experiments (to be shown later), a sampling rate of $5\%$ can provide a good accuracy with low performance overhead.
Therefore, we choose $5\%$ as the default sampling rate in our implementation. Note that for PBT, the prediction over the sampled data points is actually based on their \textit{original real neighbors} instead of their neighbors in the sampled points. Thus, the sampling process for PBT will not introduce additional errors.

\section{Compression Quality Estimation of Lossy Data Reduction in Stage II}
\label{sec:estimation}

In this section, we provide an in-depth analysis of the lossy data reduction in Stage II. We then propose a general estimation model to predict the compression quality (including compression ratio and compression error) accurately for lossy compressors with vector quantization or embedded coding in Stage II based on the theorems derived in Section \ref{sec:analysis}. After that, we apply our estimation model specifically to SZ and ZFP to predict their compression quality (i.e., Step 2 as shown in Figure \ref{fig:workflow}). We then discuss the implementation details of our proposed online automatic selection algorithm. 

For notation, we use $X^{(2)} = \{x_1^{(2)}, \cdots, x_N^{(2)}\}$ to denote the transformed data from the original data $X$. That is, $X^{(2)}$ are the input data of Stage II and the output data of Stage I.

\subsection{Estimation Based on Static Quantization}
Unlike data-dependent quantization approach (such as embedded coding that will be discussed later), static quantization is determined before performing quantization. Vector quantization \cite{vq} is one of the most popular static quantization methods. It converts $X^{(2)}$ (i.e., prediction errors in PBT or transformed data in BOT) to another set of integer values, which are easier to compress. 
Specifically, the value range is split into multiple intervals (i.e., quantization bins) based on some method, such as equal-size quantization or log-scale quantization (discussed later). Then, the compressor needs to go through all the transformed data ($X^{(2)}$) to determine in which bins they are located, and represents their values by the corresponding bin indexes, which are integer values.
During the decompression, the midpoint of each quantization bin will be used to reconstruct the data that are located in the bin; it is called the estimated value (or quantized value) in the following discussion.
The effectiveness of the data reduction in vector quantization depends on the distribution of the transformed data $X^{(2)}$.
Moreover, the quantization step introduces errors to $X^{(2)}$, and such errors will be added to the decompressed data.

\begin{figure}
\centering
\includegraphics[scale=0.46]{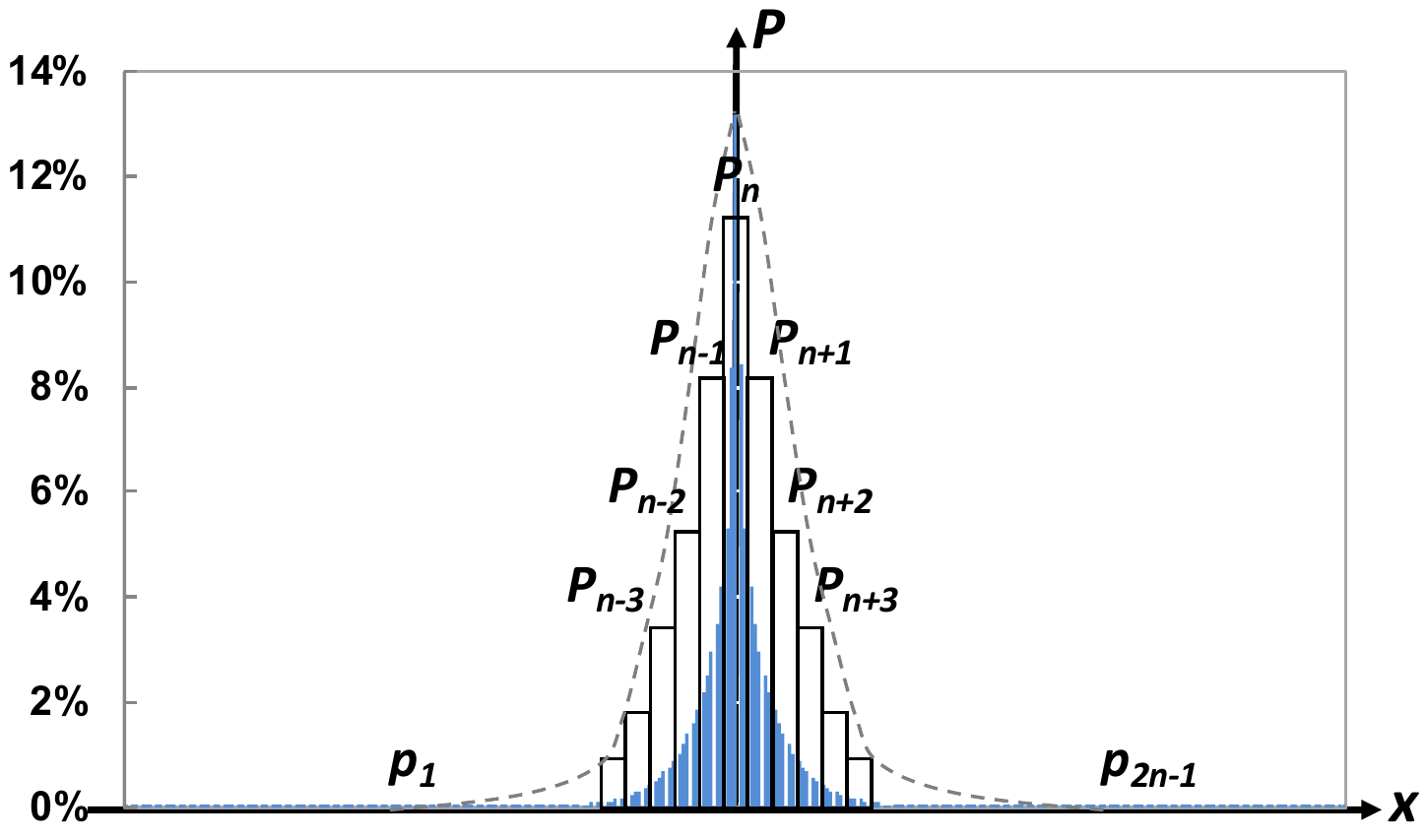}
\caption{Example of the distribution and vector quantization of the prediction errors generated by SZ lossy compressor on one ATM field.}
\label{fig:f1}
\vspace{-4mm}
\end{figure}

We build a model to estimate the data reduction level (e.g., compression ratio) and the data distortion level (e.g., mean squared error), based on a vector quantization method with a specific distribution of $X^{(2)}$.
In the following, we define some important notations to be used.
We denote $P(x)$ as the probability density function (PDF) of $X^{(2)}$, that is, $X^{(2)} \sim P(x)$.
Based on our observation, the probability distribution of $X^{(2)}$ is symmetric in a large majority of cases.
The blue area in Figure \ref{fig:f1} exemplifies the typical probability distribution of the prediction errors generated by the SZ lossy compressor using the ATM data set. All other tested data sets show the same symmetry.
Therefore, we assume $P(x)$ to be symmetric without loss of generality (i.e., $P(i)$ is equal to $P(2n-i)$ as illustrated in Figure \ref{fig:f1}), and the number of vector quantization bins is represented by $2n-1$.
We denote $\delta_i$ the length of the $i$th quantization bin, where $\delta_i = \delta_{2n-i}$ because of the symmetry property.

\subsubsection{Estimation of bit-rate}
Bit-rate is defined as the average number of bits used in the compressed data as per value. As discussed previously, a large number of transformed data values generated by the vector quantization are supposed to gather in a few quantization bins. That is, they are represented by a few integer bin indexes (Stage II), such that the data size can be reduced significantly by entropy encoding (Stage III). We combine our discussion for Stage II and Stage III, in order to estimate the overall reduction size achieved by the quantization.

Given a number of symbols, an entropy encoding method (such as Huffman coding \cite{huffman} and arithmetic coding \cite{arithmetic}) can assign a number of bits to represent these symbols based on their frequencies.
Since Shannon entropy theory \cite{shannon} gives the expected number of bits to represent these symbols, we can use the entropy value of the $2n-1$ bins to estimate the expected bit rate used to represent all quantized values.
The estimation equation is shown as follows:
\begin{equation}
%\resizebox{.95\hsize}{!}{$BR = - \sum\limits_{i=1}^{2n-1} P_i \cdot \log_2 P_i = - 2\sum\limits_{i=1}^{n-1} P_i \cdot \log_2 P_i - P_n \cdot \log_2 P_n$},
\resizebox{.38\hsize}{!}{$BR = - \sum\limits_{i=1}^{2n-1} P_i \cdot \log_2 P_i$},
\label{eq:entropy}
\end{equation}
where $P_i$ is the probability of the $i$th quantization bin.

The probability of each bin can be calculated by the integral of its probability density function value.
Specifically, $P_i = \int_{s_i}^{s_{i+1}}P(x)dx$, where $[s_i, s_{i+1})$ is the $i$-th quantization bin and $s_{i+1} - s_i = \delta_i $.
Since the integral is relatively complex to compute, we use $\delta_i \cdot P(\frac{s_i+s_{i+1}}{2})$ to approximate $\int_{s_i}^{s_{i+1}}P(x)dx$. Therefore, the estimation of the bit rate based on the $X^{(2)}$'s PDF is
\begin{equation*}
%\resizebox{\hsize}{!}{$BR = - 2\sum\limits_{i=1}^{n-1} \delta_i P(\frac{s_i+s_{i+1}}{2}) \cdot \log_2 (\delta_i P(\frac{s_i+s_{i+1}}{2})) - \delta_n P(0) \cdot \log_2 (\delta_n P(0))$}.
\resizebox{0.72\hsize}{!}{$BR = - \sum\limits_{i=1}^{2n-1} \delta_i P(\frac{s_i+s_{i+1}}{2}) \cdot \log_2 (\delta_i P(\frac{s_i+s_{i+1}}{2}))$}.
\end{equation*}

To further simply the equation, let $m_i$ denote the midpoint of the $i$th bin, namely, $(s_i+s_{i+1})/2$.
Then we have
\begin{align}
%&\resizebox{.62\hsize}{!}{$BR = - 2 \sum\limits_{i=1}^{n-1} \delta_i P(m_i) \cdot \log_2 (\delta_i P(m_i) ) - P(0) \cdot \log_2 P(0)$} \nonumber \\
&\resizebox{.58\hsize}{!}{$BR = - \sum\limits_{i=1}^{2n-1} \delta_i P(m_i) \cdot \log_2 (\delta_i P(m_i)$} \nonumber \\
& \resizebox{.82\hsize}{!}{$= -\sum\limits_{i=1}^{2n-1} P(m_i) \delta_i \log_2 \delta_i - \sum\limits_{i=1}^{2n-1} \delta_i P(m_i) \log_2 P(m_i))$}, \label{eq:bit-rate-pdf}
\end{align}
where $m_i = (s_i+s_{i+1})/2 = \sum_{j=1}^{i-1} \delta_j + \delta_i/2$.
Note that the midpoint of the $n$th bin is $0$ (i.e., $m_n = 0$) according to the symmetry property.

Therefore, we can estimate the bit rate value by Equation (\ref{eq:bit-rate-pdf}) given the probability density function of $X^{(2)}$.
(We discuss our method to estimate the $X^{(2)}$'s PDF in detail later.) Note that the compression ratio can be calculated by dividing the number of bits per floating-point value by the bit-rate, for example, 32/bit-rate for single-precision data and 64/bit-rate for double-precision data.

\subsubsection{Estimate of compression error}
As proved in Theorem \ref{thm:2} and \ref{thm:1}, the PBT and BOT are both $L^2$-norm-preserving transformations.
Thus, the $L^2$-norm-based error, such as the mean squared error (MSE), introduced by Stage II stays unchanged after decompression.
Therefore, we can estimate the $L^2$-norm based compression error by estimating the error of Stage II.

We denote $\hat{X}^{(2)}$ as the quantized values of $X^{(2)}$.
The MSE between $X^{(2)}$ and $\hat{X}^{(2)}$ can be calculated by
\begin{align}
& \resizebox{.55\hsize}{!}{$MSE(X^t,  \tilde{X^t}) = E_{X^t}[(X^t - \tilde{X^t})^2]$} \nonumber \\
& \resizebox{.42\hsize}{!}{$ = \int_{-\infty}^{+\infty} (x - \tilde{x})^2 \cdot P(x) dx$}, \label{eq:mse-pdf}
\end{align}
where $E[\cdot]$ represents the expectation.
Note that $\hat{x}$ is a step function, since the values in each bin are quantized to the same value.
Lossy compressors such as NUMARACK\cite{numarck}, SSEM\cite{ssem}, and SZ\cite{sz17} often use the midpoint of the quantization bin to approximate the values located in it.
Therefore,  $\hat{x} = \frac{s_i+s_{i+1}}{2} = m_i$ when $ s_i \leq x < s_{i+1}$.
We can further estimate the MSE based on the probability density function $P(x)$ and the step function $\hat{x}$ as follows:
\begin{align*}
& \resizebox{.58\hsize}{!}{
$MSE = \sum\limits_{i=1}^{2n-1} \int_{s_{i}}^{s_{i+1}} (x - \hat{x})^2 \cdot P(x) dx$} \\
& \resizebox{.55\hsize}{!}{$\approx  \sum\limits_{i=1}^{2n-1} (P(m_i) \cdot \int_{s_{i}}^{s_{i+1}} (x - m_i)^2 dx)$} \\
& \resizebox{.75\hsize}{!}{$ = \sum\limits_{i=1}^{2n-1} (P(m_i) \cdot \int_{0}^{\delta_i} (x - \frac{\delta_i}{2})^2 dx) = \frac{1}{12} \sum\limits_{i=1}^{2n-1} \delta_i^3 P(m_i)$}
\end{align*}

After that, we can calculate the root mean squared error (RMSE), normalized root mean squared error (NRMSE), and peak signal-to-noise ratio (PSNR) as follows:
\begin{align}
& \resizebox{.75\hsize}{!}{$NRMSE 	= \frac{\sqrt{MSE}}{VR} = (\sum\limits_{i=1}^{2n-1} \delta_i^3  P(m_i))^{\frac{1}{2}} / (2\sqrt{3} \cdot VR$)} \nonumber \\
& \resizebox{.5\hsize}{!}{$PSNR 	= -20 \cdot \log_{10} (NRMSE)$} \nonumber \\
& \resizebox{.85\hsize}{!}{$ = -10 \cdot (\log_{10} (\sum\limits_{i=1}^{2n-1} \delta_i^3  P(m_i)) - 2\cdot \log_{10} VR - \log_{10} 12)$},
\label{eq:psnr-pdf}
\end{align}
where $VR$ represents the value range of the original data $X^{(1)}$.
Thus far, we have established the estimation equations for bit rate and $L^2$-norm based compression error. We are now ready to derive the estimation of the most significant metric: rate-distortion.

\subsubsection{Estimation of rate-distortion}
Rate-distortion is an important metric to compare different lossy compressors, such as fixed rate lossy compressors (e.g., ZFP) and fixed accuracy lossy compressors (e.g., SZ).
For fair comparison, people usually plot the rate-distortion curve for the different lossy compressors and compare the distortion quality with the same rate. Generally, the higher the rate-distortion curve, the better the lossy compression quality.
Here the term ``rate'' means bit rate in bits/value, and ``distortion'' usually adopts PSNR.

Based on the estimation of bit rate and PSNR proposed above (i.e., Equations (\ref{eq:bit-rate-pdf}) and (\ref{eq:psnr-pdf})),
the rate-distortion depends only on $\delta_1, \delta_2, \cdots, \delta_{2n-1}$, given the probability distribution of $X^{(2)}$.
However, it is difficult to optimize the $2n-1$ values $\{\delta_i\}_{i = 1}^{2n-1}$ for the rate-distortion during the preparation stage, even if the probability distribution is classic distribution, such as Gaussian distribution.
In the following, we analyze three common, effective vector quantization approaches; the analysis can be extended by including more vector quantization methods.

\subsubsection{Detailed analysis of three vector quantization Cases}
\begin{itemize}[topsep=4pt]
\setlength\itemsep{0.5em}
\item \textbf{Linear quantization:} This is the simplest yet effective vector quantization approach, which is adopted by SZ lossy compressor.
Under this approach, all quantization bins have the same length, (i.e.,  $\delta_1  = \cdots = \delta_{2n-1} = \delta$).
On the other hand, the $2n-1$ quantization bin can cover all the prediction errors as long as the number of bins is large enough, hence, $\sum\limits_{i=1}^{2n-1} P(m_i)  \approx \frac{1}{\delta}$.
So, Equations (\ref{eq:bit-rate-pdf}) and (\ref{eq:psnr-pdf}) can be simplified as follows:
\begin{align}
& \resizebox{.67\hsize}{!}{$BR_{sz} 	= -\delta \sum\limits_{i=1}^{2n-1} P(m_i) \log_2 P(m_i) - \log_2 \delta, \label{eq:bitrate-sz}$} \\
& \resizebox{.67\hsize}{!}{$PSNR_{sz}  = 20\cdot \log_{10} (VR/\delta) + 10\cdot \log_{10} 12$} \label{eq:psnr-pdf-linear}.
\end{align}
Equation (\ref{eq:psnr-pdf-linear}) tells us that the PSNR depends only on the unified quantization bin size regardless of the distribution of transformed data from Stage I.
For example, the SZ lossy compressor sets the bin size $\delta$ to twice the absolute error bound (i.e., $eb_{abs}$) to make sure the maximum pointwise compression error within $eb_{abs}$.
So, based on Equation (\ref{eq:psnr-pdf-linear}), our PSNR estimation for SZ lossy compressor becomes
\begin{equation}
\label{eq:psnr-sz}
\resizebox{.75\hsize}{!}{$PSNR_{sz} = -20\cdot \log_{10} (eb_{abs}/VR) + 10\cdot \log_{10} 3$}.
\end{equation}
Note that $eb_{abs}/VR$ is the value-range-based relative error bound \cite{sz17} (denoted by $eb_{rel}$) defined by SZ. Unlike the pointwise relative error that is compared with each data value, value-range-based relative error is compared with value range of each data field. Thus our model can estimate the SZ's PSNR precisely based on the value-range-based relative error bound.

\item \textbf{Log-scale quantization:} Log-scale quantization is an alternative to the linear quantization, and its bin sizes follows a logarithm distribution.
Suppose one is using $2n-1$ bins to quantize $X^{(2)}$, in order to cover the maximum absolute value in $X^{(2)}$, $b$ is chosen to be $\lceil \log_{n} ( \max\limits_{i=1}^{N}\{|x^{(2)}_i|\}) \rceil$.
If $x^{(2)}_i < 0$, $x^{(2)}_i$ falls into the $n-\lfloor \log_b (-x^{(2)}_i) \rfloor$-th bin; if $x^{(2)}_i = 0$,  $x^{(2)}_i$ falls into the $n$th bin; if $x^{(2)}_i > 0$, $x^{(2)}_i$ falls into the $n+\lfloor \log_b x^{(2)}_i \rfloor$-th bin.
Thus, the log-scale quantization uses $\delta_{n-i} = b^i - b^{i-1}$, $\delta_n = 2b$, $\delta_{n+i} = b^i - b^{i-1}$ as the bin size where $1 \leq i \leq n-1$. 
Compared with linear quantization, log-scale quantization usually a has higher PSNR but a lower compression ratio.
The reason is that log-scale quantization assigns a larger number of finer bins to the high-frequency (central) regions.
Thus, according to Equation (\ref{eq:mse-pdf}), log-scale quantization's PSNR can be higher than linear quantization.
On the other hand, the distribution of the interval frequencies with log-scale quantization is more even than with linear quantization, leading to a poor entropy encoding result. Hence, for various data, it is hard to tell directly which quantization method is better in terms of rate-distortion.
The most effective way is to compare their rate-distortion estimations.

\item \textbf{Equal-probability quantization:} This vector quantization approach is employed by the NUMARCK lossy compressor.
This method generates equal probability for each quantization interval; hence, $\delta_i \cdot P(m_i) \approx \frac{1}{2n-1}$.
In this case, the estimation of bit rate equals $\lceil \log_2 (2n-1) \rceil = 1 + \log_2 n$. It shows that the entropy encoding has no effect on the $2n-1$ intervals with the same frequency.
The PSNR estimation will be $ - 10 \cdot \log_{10} (\sum\limits_{i=1}^{n} \delta_i^2 ) + 20 \cdot \log_{10} VR +10 \cdot \log_{10} (6n) $.
The bin size can be estimated by the clustering-based approximation approach (proposed by Chen et al. \cite{numarck}), such as the K-means cluster algorithm, whose time overhead is expensive.

\end{itemize}

\subsection{Estimation Based on Dynamic Quantization}
Dynamic quantization is a data-dependent manner and encodes the data progressively. For example, embedded coding (EC) \cite{ec} is the most commonly used dynamic quantization approach. It is the most important part of the BOT-based lossy compressors, such as JPEG2000 and ZFP. It generates a stream of bits that are put into order based on their impact on the error. Many variances of EC have been proposed in previous work \cite{lian2003analysis, fang2005parallel, chiang2004high}.
As we proved in Section \ref{sec:bot}, the $L^2$-norm based compression error in the original data space is equal to the compression error in the transformed space, so the bits in the same bit-plane (as shown as the blue dash line in Figure \ref{fig:f2} should be encoded at a time, that is the case for most of the EC variances.

Figure \ref{fig:f2} shows an example with 16 transformed (4$\times$4 data block) to be encoded. Each datum is represented by its binary format. EC starts from the leftmost bit-plane and ends at the maximum bit-plane, as shown as the purple dash line in the figure. The maximum bit-plane is determined by the bit budget or the error bound set by users. For each value, EC encodes only its significant bits (i.e., nonzero bits). We use a red dashed line to indicate the significant bits for the 16 values. The BOT-based transformed data is roughly in order (i.e., large values appear ahead of small values). We can observe that the red dashed line in Figure \ref{fig:f2} exhibits a staircase shape. We can use this feature to estimate the bit-rate and compression error for embedded coding.

\begin{figure}
\centering
\includegraphics[scale=0.32]{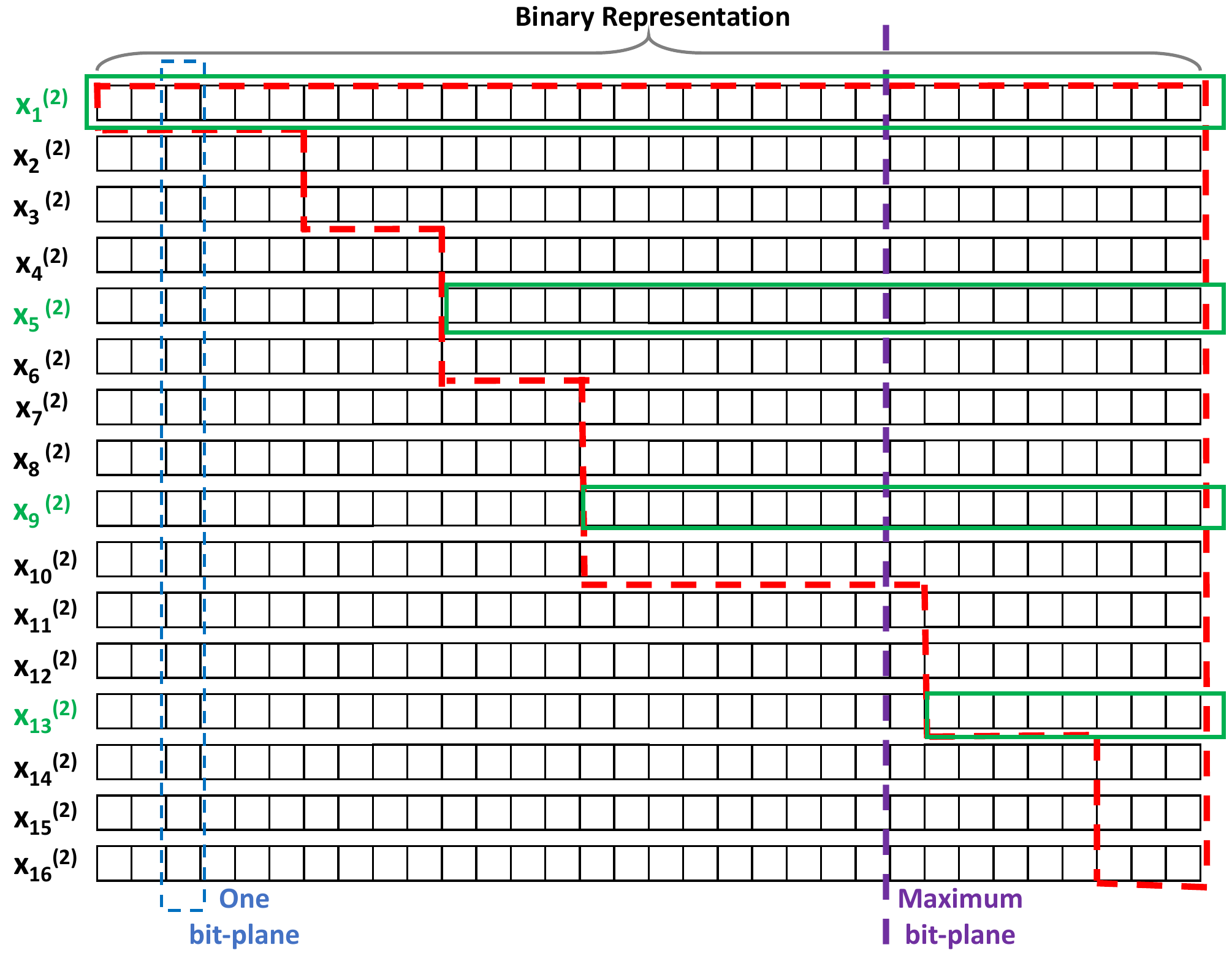}
\caption{Illustration of embedded coding scheme used in lossy compression.}
\label{fig:f2}
\vspace{-4mm}
\end{figure}

\subsubsection{Estimation of bit-rate}
To estimate the bit-rate after embedded coding, we need to estimate the number of significant bits (denoted by $n_{sb}$) for each value.
We also use a sampling approach to make an estimation.
Specifically, we first sample some data points (marked in green in the figure) and count their $n_{sb}$.
After that, we use these sampled data points and their $n_{sb}$ to interpolate $n_{sb}$ for the remaining data points.
We then calculate the average of all $n_{sb}$ (denoted by $\bar{n}_{sb}$) and use it as the approximate bit-rate value, that is, $\bar{n}_{sb}$.
The key reason we can adopt sampling and interpolation is that the significant bits of BOT-based transformed data exhibit a staircase shape, as discussed previously.

\subsubsection{Estimation of compression error}
Similar to the estimation of bit-rate, we can estimate the compression error, for example, MSE by calculating the MSE of all the sampled data points (denoted by $MSE_{sp}$). 
Note that after the first step of exponent alignment, different blocks may have different exponent offsets and maximum bit-planes. Hence, in order to calculate the overall MSE for all the sampled data points, each sampled data point's error is calculated by multiplication of its truncated error in binary representation and its block's exponent offset value. Finally, we can estimate the overall PSNR by the PSNR of the sampled data points: $PSNR_{sp} = -10 \cdot \log_{10} MSE_{sp} + 20 \cdot \log_{10} VR$.

We use $r_{sp}$ to denote the sampling rate in Stage I and $r_{sp}^{ec}$ to denote the sampling rate used in embedded coding.
We observe that low $r_{sp}^{ec}$ may significantly affect the estimation accuracy, but the estimation accuracy is not that sensitive to $r_{sp}$. Thus, as the default setting of our solution, we sample 3 data points for one 1D data block, 9 data points for one $4\times 4$ 2D data blocks, and 16 data points for one 3D $4\times 4 \times 4$ data block.
We adopt a low rate for sampling the data blocks such that our estimation model can achieve both high estimation accuracy and low overhead (illustrated later).

\subsection{Implementation of Proposed Online Selection Method}
We develop the automatic online selection method based on our proposed estimation model for two leading error-controlled lossy compressors.
As discussed above, we apply our estimation model to both SZ and ZFP to predict their compression quality accurately. Specifically, SZ adopts a multidimensional prediction model for its PBT in Stage I and linear quantization for its vector quantization in Stage II. We use Equations (\ref{eq:bitrate-sz}) and (\ref{eq:psnr-sz}) to predict its bit-rate and PSNR, respectively.
ZFP uses an optimized orthogonal transformation for its BOT in Stage I and group-testing-based EC \cite{zfp} for its EC in Stage II. We use $\hat n_{sb}$ and $PSNR_{sp}$ to predict its bit-rate and PSNR, respectively.

\begin{algorithm}
\caption{Proposed automatic online selection method for lossy compression of HPC scientific data sets}
\footnotesize
\renewcommand{\algorithmiccomment}[1]{/*#1*/}
\begin{flushleft}
\textbf{Input}:  Data fields $\{X_i\}_{i=1}^m$ to compress, user-set error bound $eb_{abs}$ or $eb_{rel}$, sampling rate $r_{sp}$ and $r_{sp}^{ec}$.\\
\textbf{Output}: Compressed-byte stream $\{C_i\}_{i=1}^m$ with selection bits $\{s_i\}_{i=1}^m$.
\end{flushleft}
\begin{algorithmic}[1]
\label{alg:alg1}
\FOR {each data field $X_i$ ($i = 1, \cdots, m$)}
\STATE Set $eb = eb_{abs}$ or $eb_{rel}\cdot VR$ (where $VR$ is the value range of $X_i$)
\STATE Sample data points from set $X_i$ blockwise to form subset $X_{sp}^{blk}$ with sampling rate $r_{sp}$
\STATE Sample data points from subset $ X_{sp}^{blk}$ pointwise to form subset $X_{sp}^{ec}$ with sampling rate $r_{sp}^{ec}$
\STATE Estimate bit-rate of ZFP (i.e., $BR_{zfp}$) by $\bar n_{sb}$ based on $X_{sp}^{ec}$ and $eb$
\STATE Estimate PSNR of ZFP (i.e., $PSNR_{zfp}$) by $PSNR_{sp}$
\STATE Calculate bin size $\delta$ based on $PSNR_{zfp}$ and Equation (\ref{eq:psnr-pdf-linear}) with $PSNR_{sz} = PSNR_{zfp}$
\STATE Construct approximate probability density function $P(\cdot)$ based on sampled data $X_{sp}^{blk}$
\STATE Estimate bit-rate of SZ (i.e., $BR_{sz}$) by Equation (\ref{eq:bitrate-sz}) based on $P(\cdot)$ and $\delta$
\IF {$BR_{sz} < BR_{zfp}$}
\STATE Perform SZ compression on $X_i$ with absolute error bound $2\cdot \delta$   
\ELSE
\STATE Perform ZFP compression on $X_i$ with absolute error bound $eb$
\ENDIF
\STATE Output compressed bytes $C_i$ and selection bit $s_i$ (e.g., $s_i = 0$ represents for SZ, $s_i = 1$ represents for ZFP)
\ENDFOR
\end{algorithmic}
\end{algorithm}

Our online selection method adopts the rate-distortion as a criterion to select the best-fit compression technique between SZ and ZFP. Specifically, for each field/variable, our solution first estimates ZFP's bit-rate and PSNR based on a given error-bound set by users. Next, it estimates SZ's bit-rate based on the PSNR estimated for ZFP, due to the high PSNR estimation accuracy in our model. Then, it selects the bestfit compressor with smaller bit-rate estimated and performs the corresponding lossy compression, as shown in Algorithm \ref{alg:alg1}. Note that the compression errors of ZFP follow a Gaussian-like distribution while those of SZ follow an uniform-like distribution \cite{peter-distribution}. Thus, in order to keep the same PSNR, ZFP needs a larger error bound as an input than does SZ. Accordingly,  with $PSNR_{sz} = PSNR_{zfp}$ (line 7), the calculated absolute error bound for SZ (i.e., $2\cdot \delta$) is smaller than the absolute error bound (i.e., $eb_{abs}$), which can guarantee the compression errors to be still bounded by $eb_{abs}$ point-wise after decompression.

%% file: tex/evaluation.tex
\section{Evaluation Results and Analysis}
\label{sec:exper-eval}
In this section, we first describe the experimental platform and the HPC scientific data sets used for evaluation. We then evaluate the accuracy of our estimation model and analyze the time and memory overhead of our online selection method. We then present the experimental results based on a parallel environment with up to 1,024 cores.

\subsection{Experimental Setting and Scientific Simulation Data}
We conduct our experimental evaluations on the Blues cluster \cite{cluster} at Argonne Laboratory Computing Resource Center using 1,024 cores (i.e., 64 nodes, each with two Intel Xeon E5-2670 processors and 64 GB DDR3 memory, and each processor with 16 cores). The storage system uses General Parallel File Systems (GPFS).
These file systems are located on a raid array and served by multiple file servers. The I/O and storage systems are typical high-end supercomputer facilities. 
We use the file-per-process mode with POSIX I/O \cite{posix} on each process for reading/writing data in parallel \footnote{POSIX I/O performance is close to other parallel I/O performance such as MPI-IO \cite{mpiio} when thousands of files are written/read simultaneously on GPFS, as indicated by a recent study \cite{parallel-io}.}.
We perform our evaluations on various single floating-point data sets including 2D ATM data sets from climate simulations, 3D Hurricane data sets from the simulation of the hurricane Isabela, and 3D NYX data sets from cosmology simulation. The details of the data sets are described in Table \ref{tab: data}.
We use SZ-1.4.11 with the default mode and ZFP-0.5.0 with the fixed accuracy mode for the following evaluations.

\begin{table}
\centering
\caption{Data Sets Used in Experimental Evaluation}
\label{tab: data}
\begin{adjustbox}{max width=\columnwidth}
\begin{tabular}{|c|c|c|c|c|}
\hline
          & Data Source          &     \# Fields & Data Size & Example Fields \\  \hline
\textbf{NYX} & Cosmology  & 6     & $147$ GB & baryon\_density, temperature \\
\hline
\textbf{ATM}       & Climate   &  79    & $1.5$ TB & CLDHGH, CLDLOW\\ \hline
\textbf{Hurricane} & Hurricane & 13      & $62.4$ GB & QICE, PRECIP, U, V, W\\  \hline
\end{tabular}
\end{adjustbox}
\vspace{-4mm}
\end{table}

\subsection{Accuracy of Compression-Quality Estimation}
We evaluate our model based on three criteria: average error of estimating bit-rate, average error of estimating PSNR, and accuracy of selecting the best-fit compression technique under different sampling rates. Note that here we use PSNR instead of MSE because previous work \cite{zfp,sz17} usually adopt PSNR for rate-distortion evaluation.

Tables \ref{tab:accuracy-atm} and \ref{tab:accuracy-hurricane} show the average errors of bit-rate and PSNR under different sampling rates (i.e., 1\%, 5\%, and 10\%).
They exhibit that our estimation model has a relatively high accuracy in estimating PSNR with low sampling rate.
For example, for both SZ and ZFP, with 5\% sampling rate, the average PSNR estimation errors are within 2\% for the ATM data sets and within 4\% for the Hurricane data sets.

As for the bit-rate estimation, the experiments based on ATM and Hurricane data sets show that the bit-rate values estimated for SZ are always lower than the real bit-rate values after compression, and the estimation error can be up to 19\% in some cases.
The reason is that our model adopts the Shannon entropy theory (i.e., Equation (\ref{eq:bit-rate-pdf})) to approximate the bit-rate for SZ. Note that the entropy value is the optimal value in theory, while the designed/implemented entropy encoding algorithm (such as Huffman encoding) may not reach such a theoretical optimum in practice.
This situation typically happens when the data set has a lot of similar values such that it is easy to compress with a high compression ratio.
%The hurricane Isabela data sets, for example, contain several easy-to-compress variables, thereby, the real bit-rates of SZ are generally higher than our estimated ones in those cases.

To address the above issue, we improve the estimation accuracy for SZ by introducing a positive offset, which is set to 0.5 bits/value based on our experiments using real-world simulation data.
Tables \ref{tab:accuracy-atm}, \ref{tab:accuracy-hurricane}, {\ref{tab:stdev-atm}, and \ref{tab:stdev-hurricane} present the accuracy (average and standard deviation of relative estimation error) of the offset-based estimation for SZ and the original estimation approach for ZFP. We can see that our final estimation model can always predict both the bit-rate and PSNR accurately for the two compressors. In relative terms, the bit-rate estimation errors fall into the interval [-8.5\%, 7.5\%] for SZ. Note that here the negative values represent that the estimated values are lower than the real values. As for ZFP, the bit-rate estimation errors are limited within 5.7\% on the ATM data sets and 0.9\% on the Hurricane data sets, when the sampling rate is higher than 5\%. 
The PSNR estimation errors vary from -3.5\% to -0.6\% when the sampling rate is set to 5\%. Hence, we suggest setting the sampling rate to 5\% in practice, which also has little time cost (presented latter).

%\textcolor{red}{(Why is the bit-rate always more accurate for ZFP and PSNR more accurate for SZ in Table 2?)}
As illustrated in Table \ref{tab:accuracy-atm} and \ref{tab:accuracy-hurricane}, our estimation of bit-rate is more accurate for ZFP than SZ in most instances. The reason is that we estimate the bit-rate by calculating the entropy value (i.e., Equation (5)) for SZ because of the Entropy encoding step (Huffman coding) adopted in SZ. As mentioned above, entropy value represents an optimal bit-rate (or lower bound) in theory, which leads to a certain estimation error.  
The tables also show that our estimation of PSNR is more accurate for SZ than ZFP under all the tested sampling rates. The key reason is that the symmetric distribution of prediction errors in SZ is not related to its prediction accuracy, but the staircase shape of transformed coefficients in ZFP is highly dependent on its transformation efficiency. Therefore, our PSNR modelling based on SZ’s quantization errors is more accurate than that based on ZFP’s truncation errors.
%\textcolor{red}{(Why is the standard deviation for ZFP in Table 4 so high?)}
We also note that the standard deviation of bit-rate error is much higher for ZFP than SZ on the ATM data sets, as shown in Table \ref{tab:stdev-atm}. This is because ZFP’s block orthogonal transformation may have low decorrelation efficiency on certain fields in the ATM data sets, so the transformed coefficients can still have high correlation and the staircase shape (as shown in Figure \ref{fig:f2}) cannot be always established on these fields, which can result in large bit-rate error fluctuations and relatively high standard deviation.
%\textcolor{red}{(Why are all the PSNR errors negative?)}
In addition, we note that since our proposed estimation of PSNR is based on the approach to control the maximum $L_2$-norm-based compression error, the estimated PSNRs are always lower than the real PSNRs, leading to the negative PSNR errors shown in Table \ref{tab:accuracy-atm} and \ref{tab:accuracy-hurricane}.
%\textcolor{red}{(Why is the SZ bit-rate error always negative in Table 3 but always positive in Table 2? What features of the data set cause this behavior?)}
Finally, it is worth noting that the Hurricane data sets have more high-compression-ratio variables than the ATM data sets. In other words, the Hurricane data sets are relatively easier to compress compared with the ATM data sets. Hence, using the entropy value (i.e., the optimal value in theory) to estimate the bit-rate is more accurate for the Hurricane data sets than for the ATM data sets. Consequently, considering the 0.5 bits/value offset for SZ, the bit-rate errors are always negative on the ATM data sets (as shown in Table \ref{tab:accuracy-atm} while positive on the Hurricane data sets (as shown in Table \ref{tab:accuracy-hurricane}).
%\textcolor{red}{(Can it be exploited to lower the error?)}
In the future work, we can further improve our estimation method by introducing the offset unless the data set is relatively hard to compress.

We next evaluate the selection accuracy, which is calculated as the ratio of the number of correct selections to the total number of variables or data sets. The correct selection means our model make a correct decision by selecting the bestfit compression technique. The selection accuracy is 98.7\% on the Hurricane data sets and 88.3\% on the ATM data sets. In fact, the 1.3\% wrong selection brings only 0.08\% compression-ratio degradation on the Hurricane data sets, and the 11.7\% wrong selection leads to only 3.3\% compression-ratio degradation on the ATM data sets, as shown in Figure \ref{fig:f5}. The reason the wrong selections leads to little degradation is that almost all the wrong selections actually happen only when the two compressors exhibit close bit-rates with the same PSNR, such that selecting either of them may not affect the final overall compression quality by much.

\begin{table}
\centering
\footnotesize
\caption{Average Relative Error of Our Estimation Model for Compression Quality on 2D ATM Data Sets}
\label{tab:accuracy-atm}
\begin{tabular}{|l|c|c|c|c|c|c|}
\hline
\multirow{2}{*}{} & \multicolumn{2}{c|}{\textbf{$r_{sp} = 1\%$}} & \multicolumn{2}{c|}{\textbf{$r_{sp} = 5\%$}} & \multicolumn{2}{c|}{\textbf{$r_{sp} = 10\%$}} \\ \cline{2-7}
                  & \textit{SZ}                  & \textit{ZFP}               & \textit{SZ}            & \textit{ZFP}            & \textit{SZ}             & \textit{ZFP}            \\ \hline
\textbf{Bit-rate} & 7.5\%               & 5.7\%                      & 7.4\%                  & 5.7\%                   & 7.3\%                   & 5.6\%                   \\ \hline
\textbf{PSNR}     & -2.5\%               & -4.1\%                      & -1.1\%                  & -2.0\%                   & -0.6\%                   & -1.6\%                   \\ \hline
\end{tabular}
\end{table}

\begin{table}
\centering
\footnotesize
\caption{Average Relative Error of Our Estimation Model for Compression Quality on 3D Hurricane data sets}
\label{tab:accuracy-hurricane}
\begin{tabular}{|l|c|c|c|c|c|c|}
\hline
\multirow{2}{*}{} & \multicolumn{2}{c|}{\textbf{$r_{sp} = 1\%$}} & \multicolumn{2}{c|}{\textbf{$r_{sp} = 5\%$}} & \multicolumn{2}{c|}{\textbf{$r_{sp} = 10\%$}} \\ \cline{2-7}
                  & \textit{SZ}                  & \textit{ZFP}               & \textit{SZ}            & \textit{ZFP}            & \textit{SZ}             & \textit{ZFP}            \\ \hline
\textbf{Bit-rate}  & -4.5\%                & 8.0\%                      & -8.5\%                   & 0.9\%                   & -4.6\%                    & 0.9\%                   \\ \hline
\textbf{PSNR}     & -2.6\%               & -6.3\%                      & -1.1\%                  & -3.5\%                   & -0.8\%                   & -3.1\%                   \\ \hline
\end{tabular}
\vspace{-4mm}
\end{table}

\begin{table}
\centering
\footnotesize
\caption{Standard Deviation of Relative Estimation Error for Compression Quality on 2D ATM Data Sets}
\label{tab:stdev-atm}
\begin{tabular}{|l|c|c|c|c|c|c|}
\hline
\multirow{2}{*}{} & \multicolumn{2}{c|}{$r_{sp} = 1\%$} & \multicolumn{2}{c|}{$r_{sp} = 5\%$} & \multicolumn{2}{c|}{$r_{sp} = 10\%$} \\ \cline{2-7} 
                  & \textit{SZ}      & \textit{ZFP}     & \textit{SZ}      & \textit{ZFP}     & \textit{SZ}      & \textit{ZFP}      \\ \hline
\textbf{Bit-rate} & 8.9\%           & 23.9\%           & 8.8\%           & 23.6\%           & 8.8\%           & 23.5\%            \\ \hline
\textbf{PSNR}     & 5.6\%            & 6.0\%            & 3.1\%            & 4.0\%            & 1.5\%            & 3.8\%             \\ \hline
\end{tabular}
\end{table}

\begin{table}
\centering
\footnotesize
\caption{Standard Deviation of Relative Estimation Error for Compression Quality on 3D Hurricane Data Sets}
\label{tab:stdev-hurricane}
\begin{tabular}{|l|c|c|c|c|c|c|}
\hline
\multirow{2}{*}{} & \multicolumn{2}{c|}{$r_{sp} = 1\%$} & \multicolumn{2}{c|}{$r_{sp} = 5\%$} & \multicolumn{2}{c|}{$r_{sp} = 10\%$} \\ \cline{2-7} 
                  & \textit{SZ}      & \textit{ZFP}     & \textit{SZ}      & \textit{ZFP}     & \textit{SZ}      & \textit{ZFP}      \\ \hline
\textbf{Bit-rate} & 10.4\%           & 11.9\%           & 16.0\%           & 2.0\%            & 10.8\%           & 3.1\%             \\ \hline
\textbf{PSNR}     & 2.2\%            & 5.1\%            & 1.2\%            & 3.3\%            & 2.0\%            & 1.0\%             \\ \hline
\end{tabular}
\end{table}

\begin{table*}[t]
\centering
\footnotesize
\caption{Average Time Overhead for One Field Compared with Compression Time of SZ and ZFP on NYX, ATM, and Hurricane Data Sets}
\label{tab: overhead}
\begin{tabular}{|l|c|c|c|c|c|c|c|c|c|}
\hline
\multirow{2}{*}{}  & \multicolumn{3}{c|}{$r_{sp} = 1\%$} & \multicolumn{3}{c|}{$r_{sp} = 5\%$} & \multicolumn{3}{c|}{$r_{sp} = 10\%$} \\ \cline{2-10} 
                   & Time (sec.)     & SZ         & ZFP        & Time (sec.)     & SZ         & ZFP        & Time (sec.)    & SZ         & ZFP        \\ \hline
\textbf{NYX}       &  $1.8 \times 10^{-2}$         &  1.4\%          &   1.2\%        &    $7.4 \times 10^{-2}$       &  5.6\%          &      4.7\%      &      $1.3 \times 10^{-1}$     &     9.8\%       &     8.4\%       \\ \hline
\textbf{ATM}       & $6.0\times 10^{-3}$    & 1.5\%    & 1.9\%    & $2.0 \times 10^{-2}$    & 4.9\%    & 6.3\%    & $3.8 \times 10^{-2}$ & 9.2\%    & 11.9\%   \\ \hline
\textbf{Hurricane} & $1.6 \times 10^{-2}$    & 1.3\%    & 1.7\%    & $7.1 \times 10^{-2}$    & 5.4\%    & 7.2\%    & $1.2 \times 10^{-1}$    & 9.2\%    & 12.5\%   \\ \hline
\end{tabular}
\vspace{-4mm}
\end{table*}

\subsection{Overhead Analysis}
Next, we analyze the overhead of our automatic online selection method with respect to both time and memory.

\subsubsection{Time overhead}
Time overhead comes from two parts: the transformation of sampled data points in Step 1 (as shown in Figure \ref{fig:workflow}) and the estimation of compression quality in Step 2 (also shown in Figure \ref{fig:workflow}).
For the first part, the overhead of sampled data transformation is scaled linearly with the sampling rate $r_{sp}$. Hence, if we assume Stage I takes a percentage (denoted by $r_{stage1}$) of the total compression time, the overhead can be expressed as $O(r_{sp} \cdot r_{stage1} \cdot N)$, where $N$ is the number of data points. For example, $r_{stage1}$ is up to 60\% based on our experiments, so the time overhead of the sampled data transformation is up to $3\%$ of SZ's compression time, under a default sampling rate of $5\%$.
For the second part, when we estimate compression quality, the time complexity is $O(n)$ with vector quantization based on Equations (\ref{eq:bit-rate-pdf}) and (\ref{eq:psnr-pdf}), where $n$ is the number of quantization bins, which is very small in general compared with the data size $N$. Hence, the time overhead complexity with embedded coding is $O(r_{sp}\cdot r_{sp}^{ec} \cdot N)$.
Therefore, the overall time overhead can be expressed as $O(r_{sp}\cdot N)$ with a low constant coefficient, i.e., $O(r_{sp}^{ec}+r_{stage1})$.

Table \ref{tab: overhead} shows the time overhead on the NYX, ATM, and Hurricane data sets compared with the compression time of SZ and ZFP. It illustrates that the time overhead scales linearly with the sampling rate, consistent with our analysis.

\subsubsection{Memory overhead}
Memory overhead results from the storage of the approximate probability density function. It can be expressed as $O(n_{pdf})$, where $n_{pdf}$ is the number of bins used to represent the PDF. Note that although $n_{pdf}$ is larger than $n$, $n_{pdf}$ is still very small compared with the data size $N$. Specifically, we use $65,535$ quantization bins (i.e., $n_{pdf} = 655,35$) in our evaluation. The dimensions of each field in the ATM and Hurricane data sets are $1800\times3600$ and $100\times500\times500$ (i.e., $N = 6.48\times10^6, 2.5\times10^7$), respectively. Therefore, the memory overheads are about 1.0\% and 0.3\% on the ATM and Hurricane data sets, respectively.

\subsection{Analysis of Adaptability between Selection Methods Based on Fixed-PSNR vs. Fixed-Maximum-Error}
We compare our proposed selection method based on fixed PSNR to the solution based on fixed maximum error (proposed by Lu et al. \cite{ipdps18}) using the NYX, ATM, and Hurricane data sets, as shown in Figure \ref{fig:f9}. Their solution simply selects the compressor with the highest compression ratio based on a fixed error bound (called \emph{selection based on error bound}). 
Unlike their work that adopted point-wise relative error bound \cite{ipdps18}, we improved their selection method by using the absolute error bound instead, since both SZ and ZFP have better rate-distortions when using absolute error bound mode rather than using pointwise relative error bound mode, as confirmed in the previous studies \cite{zfp-online, drbsd2}. Specifically, for each data field, we set the absolute error bound to $10^{-3}$ of its value range. Figure \ref{fig:f9}(a) shows that the selection method based on error bound always chooses SZ as the best-fit compressor for all the tested fields because SZ always leads to the higher compression ratios than ZFP does on these fields given a specific error bound. 
We note that ZFP over-preserves the compression error with respect to the user-set error bound. Thus, ZFP may have a higher PSNR than does SZ, even if its compression ratio is lower. Our proposed method is designed to select the compressor that has lower bit-rate (i.e., higher compression ratio) with the same PSNR (called \emph{selection based on rate-distortion}), leading to better overall rate-distortion result. Figure \ref{fig:f9}(b) shows that our method can select the different best-fit compressors based on the rate-distortion for the different fields in the tested data sets.

\begin{figure}
\centering
\includegraphics[scale=0.36]{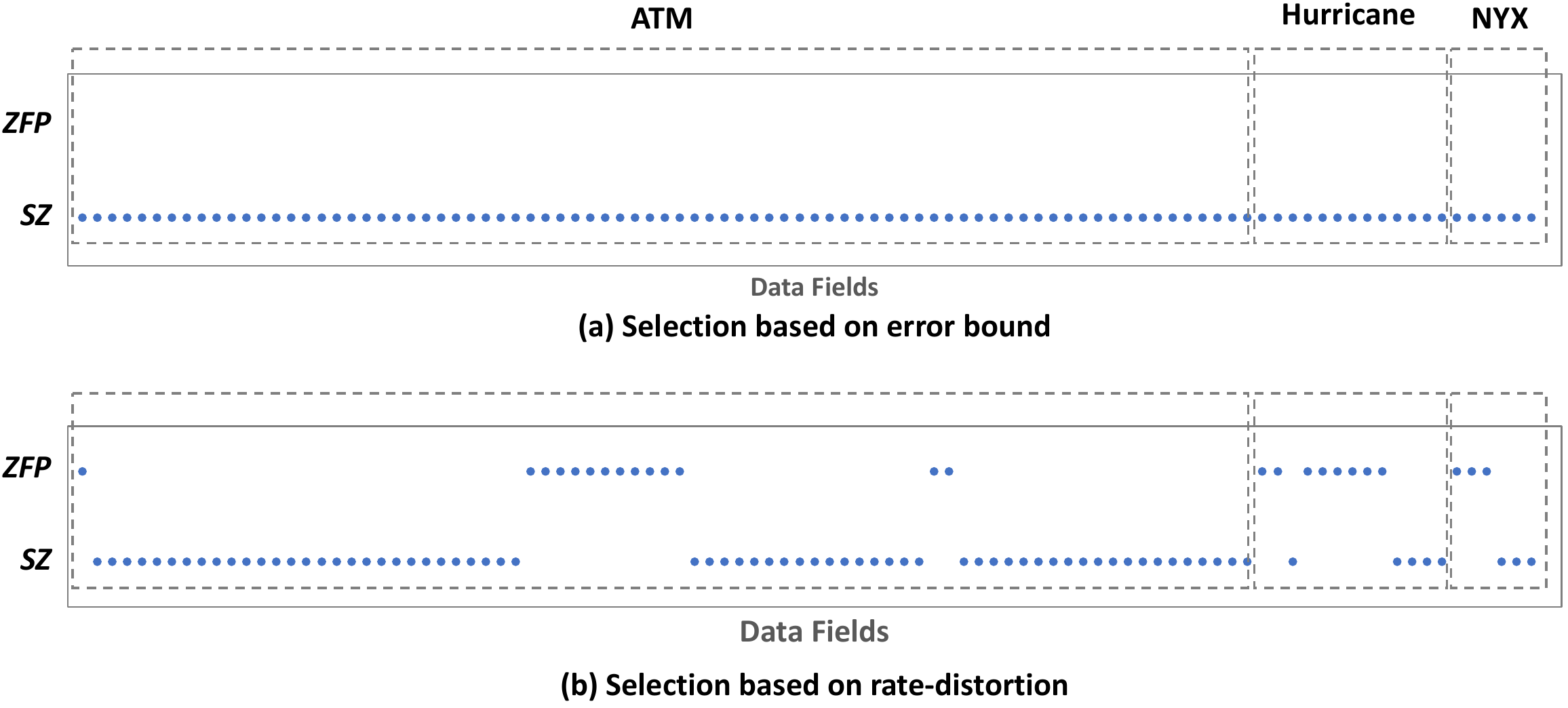}
\caption{Illustration of different selection methods, i.e., (a) selection based on error bound and (b) selection based on rate-distortion, on experimental data sets.}
\label{fig:f9}
\vspace{-4mm}
\end{figure}

\subsection{Empirical Performance Evaluation}
We evaluate the overall performance of our proposed solution in parallel. Let us first consider the compression ratio improvement achieved by our compressor. Figure \ref{fig:f5} shows that the compression ratio of SZ, ZFP, and our solution on the NYX, ATM, and Hurricane data sets with different error bounds. 
Our solution can outperform both SZ and ZFP because our online selection method attempt to select the better compression approach for each field in the data sets.
Note that the \textit{optimum bar} represents the compression ratios in an ideal case assuming that the best-fit compressors can always be selected for any fields in the data sets.
Specifically, the compression ratio of our solution outperforms that of the worst solution by 62\%, 36\%, 19\% on the Hurricane data sets, by 28\%, 38\%, 20\% on the ATM data sets, and 70\%, 17\%, 12\% on the NYX data sets with the value-range-based relative error bound $eb_{reb}$ of $10 ^ {-3}$, $10 ^ {-4}$, $10 ^ {-6}$, respectively.
We compare our solution with the worst solution because our proposed selection method can almost always select the best-fit compressor; however, a user is likely to keep using the same, but maybe the worst, compressor for all data sets.

\begin{figure}
\centering
\subfigure[{$eb_{rel}=10^{-3}$}]{\includegraphics[scale=0.52]{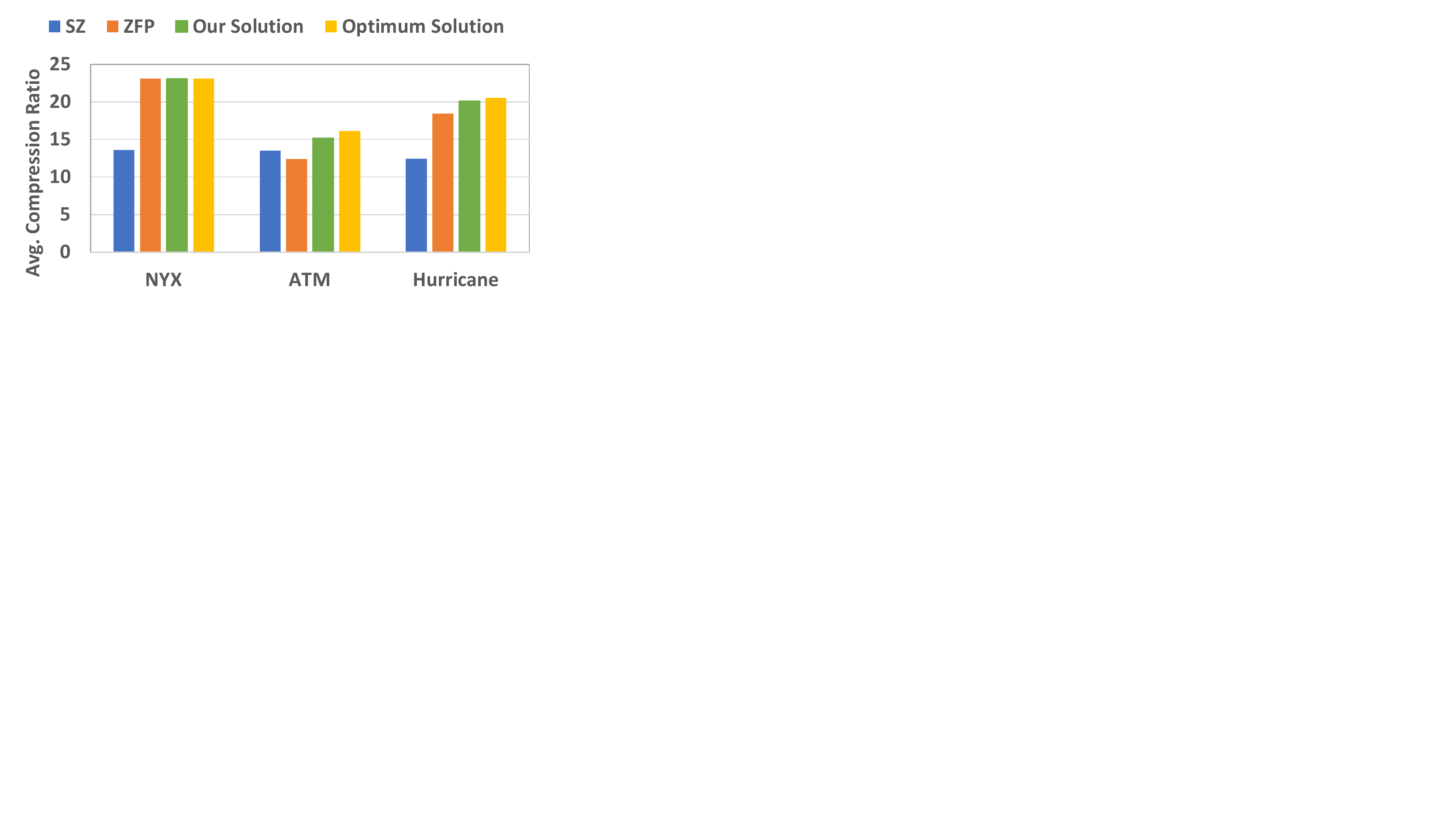}}

\subfigure[{$eb_{rel}=10^{-4}$}]{\includegraphics[scale=0.52]{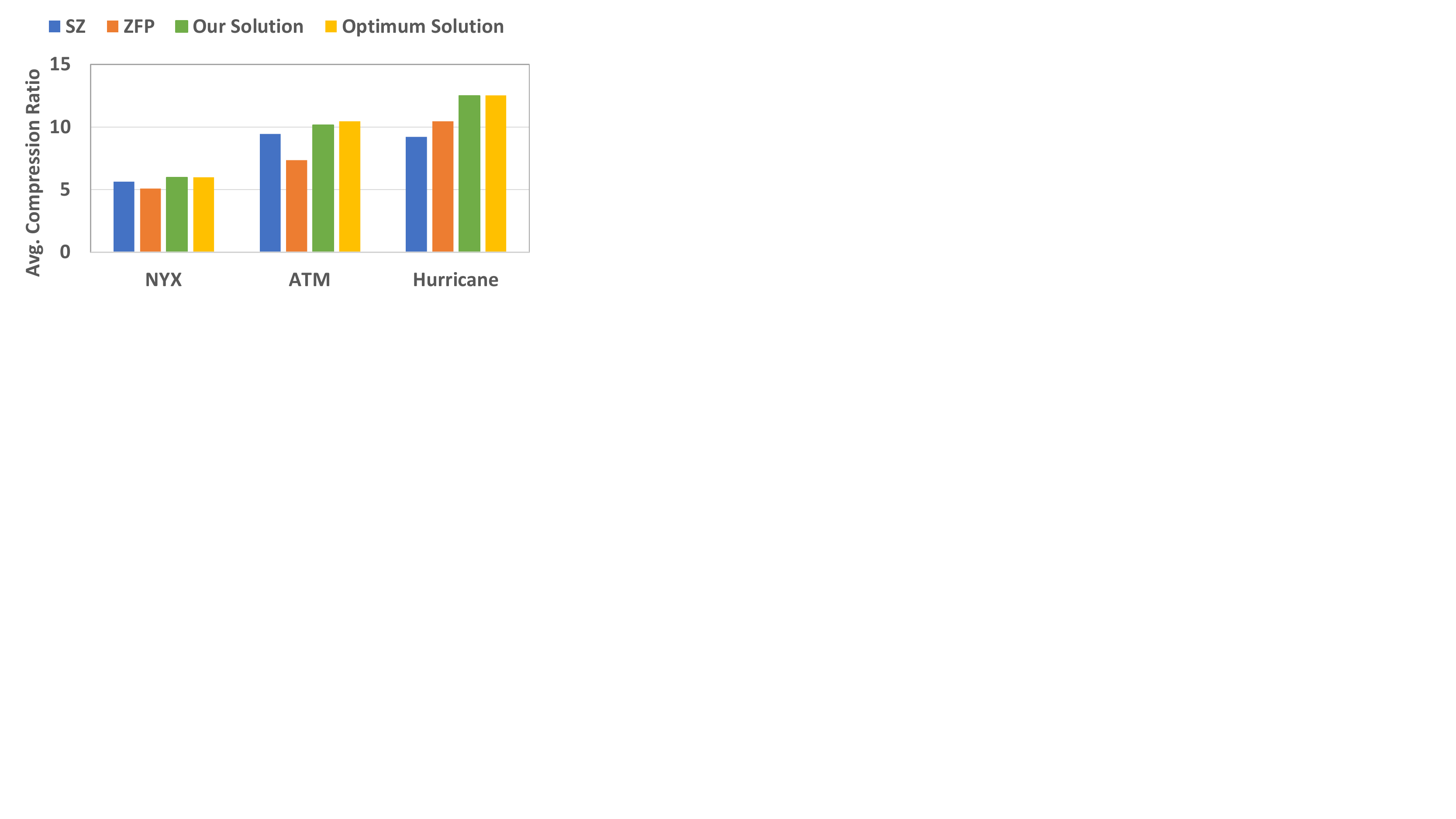}}

\subfigure[{$eb_{rel}=10^{-6}$}]{\includegraphics[scale=0.52]{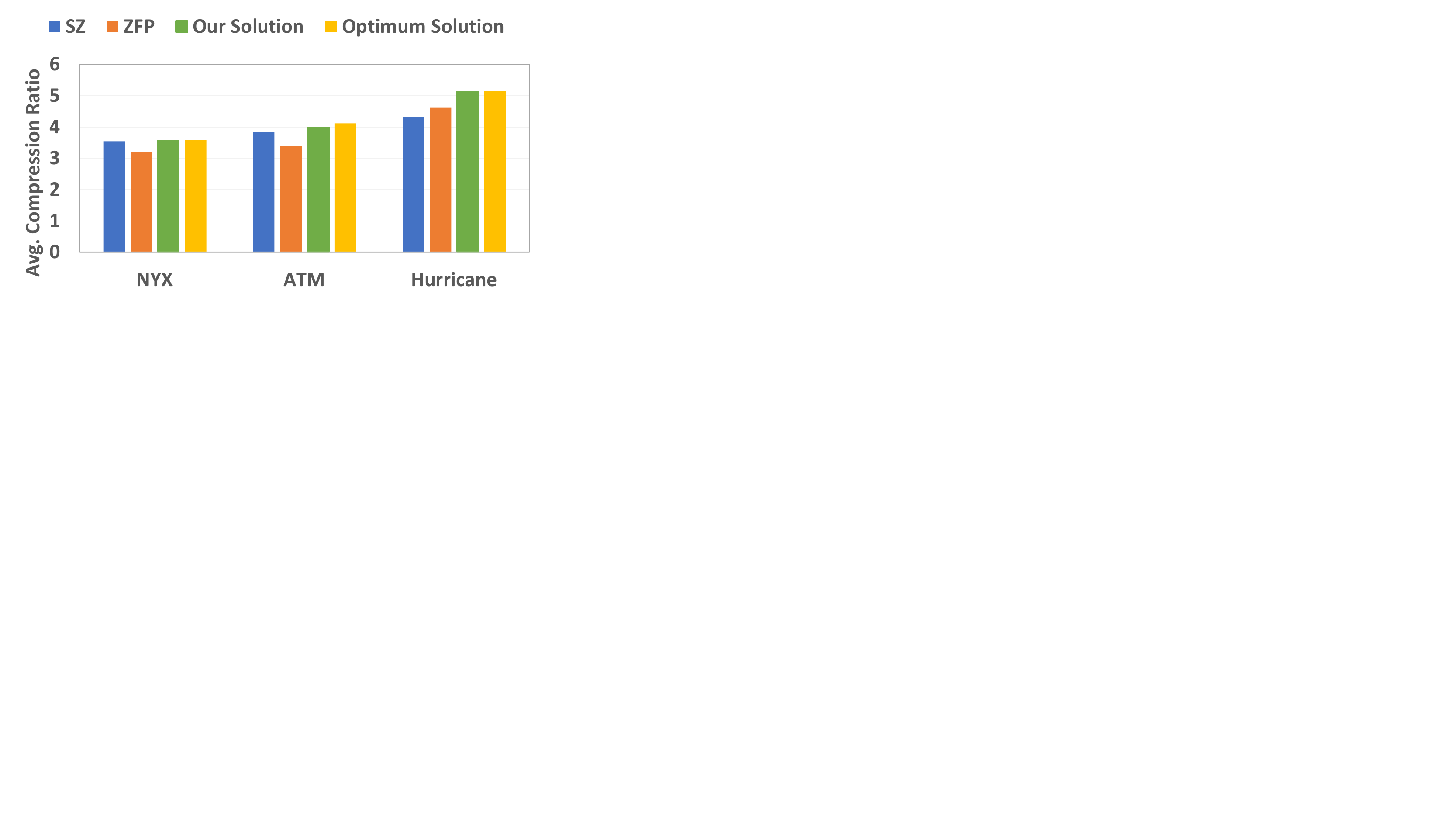}}
\caption{Average compression ratios of SZ, ZFP, and our solution on three application datasets (with the same PSNR across compressors on each field).}
\label{fig:f5}
\vspace{-6mm}
\end{figure}

In Figures \ref{fig:f6} and \ref{fig:f7}, we present the throughputs (in GB/s) of storing and loading data to GPFS with different solutions.
We increase the scale from 1 to 1,024 processes.
We set the value-range-based relative error bound $eb_{rel}$ to a reasonable value $10^{-4}$\cite{sz17}.
We test each experiment five times and use their average time to calculate the throughputs.
The storing and loading throughputs are calculated based on the compression/decompression time and I/O time.
We compare our solution with the other two solutions based on SZ and ZFP compressors and a baseline solution. The baseline solution is storing and loading the uncompressed data directly without any compression.
Figures \ref{fig:f6} and \ref{fig:f7} illustrate that our optimized compressor can achieve the highest stroing and loading throughputs compared with SZ and ZFP. Our optimized compressor can outperform the second-best solution by 68\% of the storing throughput and by 79\% of the loading throughput with 1,024 processes.
Our proposed solution has higher throughputs because it can achieve higher compression ratios than both SZ and ZFP with little extra overhead, so the time of writing and reading data is reduced significantly, leading to higher overall throughputs. 
Similarly, SZ has higher overall throughputs than ZFP does because of achieving a higher overall compression ratio on the tested data sets with the same PSNR, although the compression/decompression rates of SZ are lower than those of ZFP \cite{sz17}.
We note that the compression/decompression rates have a linear speedup with the number of processors (as illustrated in Figure 10 in \cite{sz17}) and more processes will lead to higher unexpected I/O contention and data management cost by GPFS when writing/reading data simultaneously. 
Hence, we expect that the performance gains of our solution compared with SZ and ZFP will further increase with scale because of the inevitable bottleneck of the I/O bandwidth.

\begin{figure}
\centering
 \includegraphics[scale=0.48]{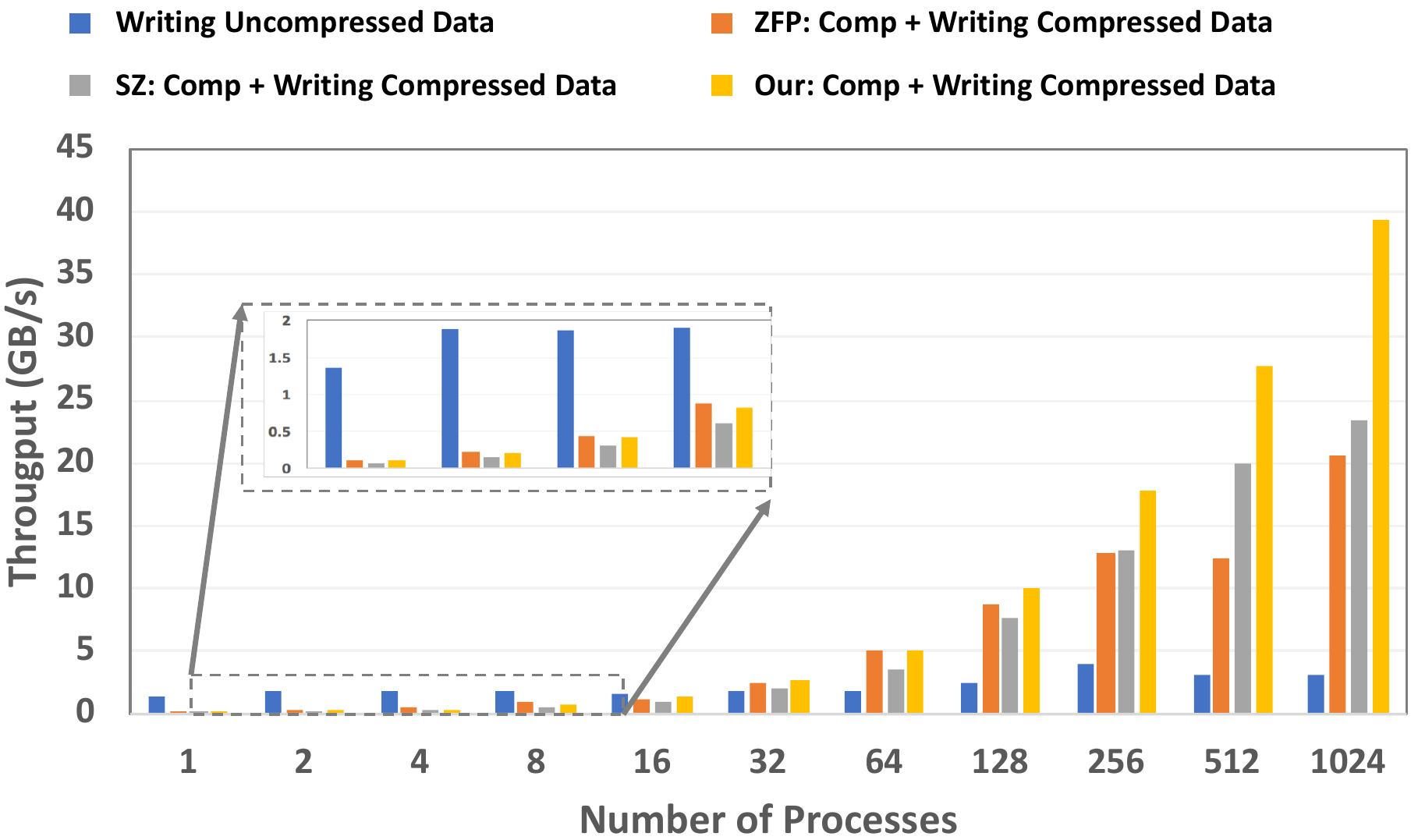}
 \caption{Throughputs of storing data (compression + I/O) with different solutions on the Hurricane data sets (with the same PSNR).}
\label{fig:f6}
\end{figure}

\begin{figure}
\centering
\includegraphics[scale=0.48]{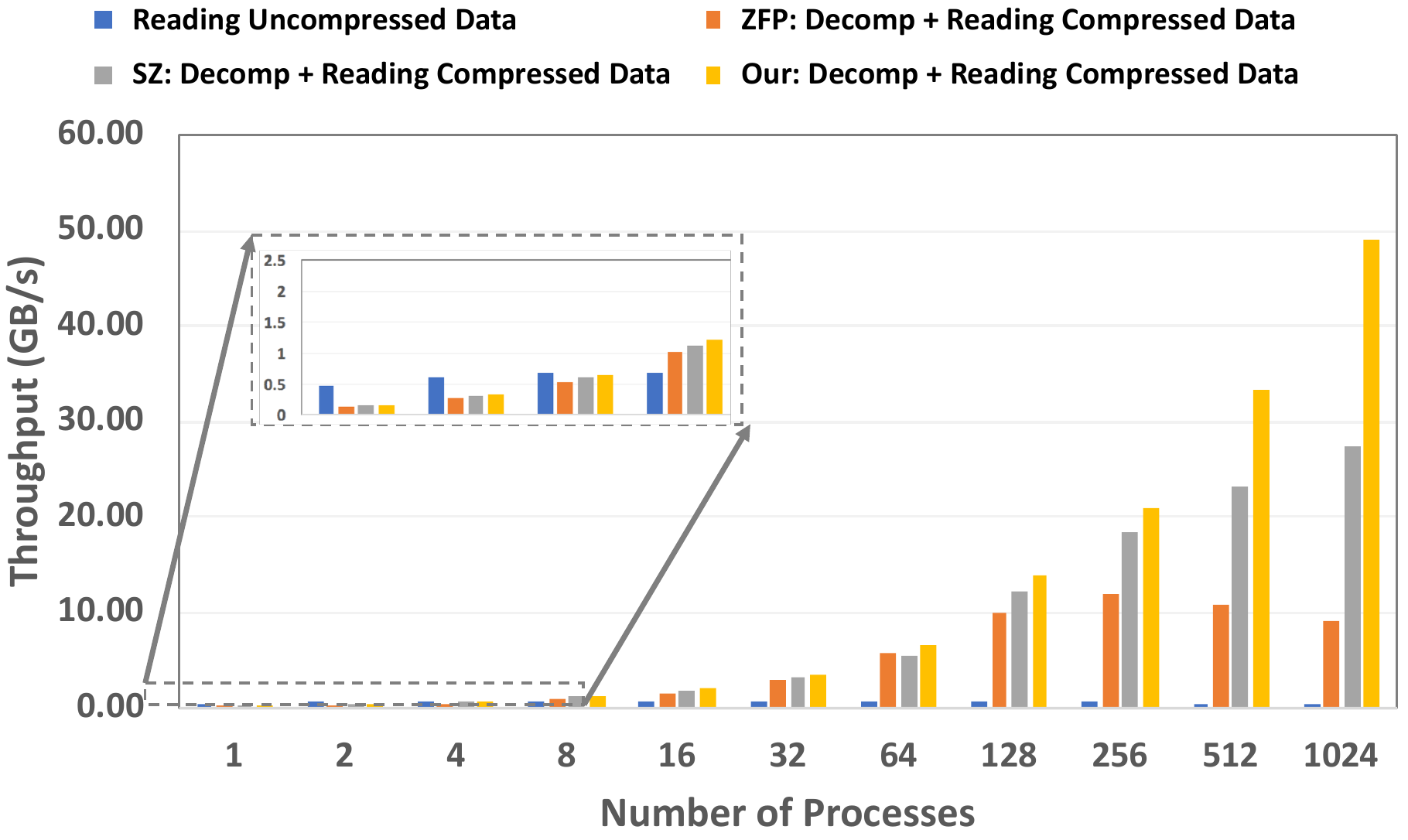}
 \caption{Throughputs of loading data (decompression + I/O) with different solutions on the Hurricane data sets (with the same PSNR).}
\label{fig:f7}
\vspace{-4mm}
\end{figure} 

%% file: tex/conclusion.tex
\section{Conclusion}
\label{sec:conclusion}

In this paper, we propose a novel online, low-overhead selection method that can select the best-fit lossy compressor between two leading compressors, SZ and ZFP, optimizing the rate-distortion for HPC data sets, This is the first attempt to derive such an approach to the best of our knowledge. We develop a generic open-source toolkit/library under a BSD license.
We evaluate our solution on real-world production HPC scientific data sets across multiple domains in parallel with up to 1,024 cores. The key findings are as follows.
\begin{itemize}
\item The average error of our estimation with default sampling rate on bit-rate (i.e., compression ratio) can be limited to within 8.5\% and 5.7\% for SZ and ZFP, respectively.
\item The average error of our estimation with default sampling rate on PSNR (i.e., data distortion) can be limited to within 1.1\% and 3.5\% for SZ and ZFP, respectively.
\item The accuracy of selecting the best-fit compressor with default sampling rate is 88.3 $\sim$98.7\%, with little analysis/estimation time overhead (within 5.4\% and 7.3\% for SZ and ZFP, respectively).
\item Our solution improves the compression ratio by 12$\sim$70\% compared with that of SZ and ZFP, with the same distortion (PSNR) of the data.
\item The overall performance in loading and storing data can be improved by 79\% and 68\% on 1,024 cores with our solution, respectively, compared with the second-best solution.
\end{itemize}

We plan to extend our optimization solution (such as estimation model) to more error-controlled lossy compression techniques, including more quantization approaches and block-based transformations, to further improve the compression qualities for more HPC scientific data sets.

%% file: tex/acknowledge.tex
\section*{Acknowledgments}
\small
This research was supported by the Exascale Computing Project (ECP), Project Number: 17-SC-20-SC, a collaborative effort of two DOE organizations - the Office of Science and the National Nuclear Security Administration, responsible for the planning and preparation of a capable exascale ecosystem, including software, applications, hardware, advanced system engineering and early testbed platforms, to support the nation's exascale computing imperative. The submitted manuscript has been created by UChicago Argonne, LLC, Operator of Argonne National Laboratory (Argonne). Argonne, a U.S. Department of Energy Office of Science laboratory, is operated under Contract No. DE-AC02-06CH11357. We acknowledge the computing resources provided by LCRC at Argonne National Laboratory.

\vspace{-2mm}